\documentclass[acmsmall,review=false,nonacm]{acmart}

\usepackage[caption=false]{subfig}
\usepackage{enumitem}
\usepackage{multirow,makecell}
\usepackage{algpseudocode,algorithm}
\usepackage{listings}
\definecolor{mGreen}{rgb}{0,0.6,0}
\definecolor{mGray}{rgb}{0.5,0.5,0.5}
\definecolor{mPurple}{rgb}{0.58,0,0.82}
\definecolor{backgroundColour}{rgb}{0.95,0.95,0.92}
\lstdefinestyle{CStyle}{
    columns=fullflexible,
    backgroundcolor=\color{backgroundColour},
    commentstyle=\color{mGreen},
    keywordstyle=\color{magenta},
    stringstyle=\color{mPurple},
    basicstyle=\ttfamily\normalsize,
    breakatwhitespace=false,
    breaklines=true,
    captionpos=b,
    keepspaces=true
    showspaces=false,
    showstringspaces=false,
    showtabs=false,
    tabsize=2,
    language=C
}

\usepackage{fontawesome} 
\usepackage{tikz}
\usetikzlibrary{graphs}
\usetikzlibrary{arrows}
\usetikzlibrary{decorations.pathreplacing}
\usetikzlibrary{decorations.pathmorphing}
\tikzstyle{stepnum}=[shade,shading=ball,ball color=green!50!orange!50!,circle,draw=black,minimum size=1.5em,inner sep=1]
\tikzstyle{stagevalue}=[fill,circle,minimum size=5pt,inner sep=1,draw]
\tikzstyle{store} = [pin edge={-to,thick,blue}]
\tikzstyle{restore} = [pin edge={to-,thick,red}]
\tikzstyle{restorefromdisk} = [pin edge={to-,thick,black},pin distance=0.3cm]
\tikzstyle{restorekeepinstack} = [pin edge={to-,thick,red,dashed}]
\tikzstyle{edgefrom} = [<-,thick]
\tikzstyle{edgebend} = [->,thick,>=stealth',bend right=60]

\newtheorem{remark}{Remark}
\newtheorem{problem}{Problem}

\algnewcommand\Assert[1]{\texttt{Assert(#1)}}

\usepackage{bm}
\newcommand{\bu}{\bm{u}}
\newcommand{\bU}{\bm{U}}
\newcommand{\N}{\bm{\mathcal{N}}}
\newcommand{\blambda}{\bm{\lambda}}
\newcommand{\f}{\bm{f}}
\newcommand{\fu}{\f_{\bu}}
\usepackage[figuresright]{rotating}
\newcommand{\REVadvance}{\textsc{advance}}
\newcommand{\REVtakeshot}{\textsc{takeshot}}
\newcommand{\REVstore}{\textsc{store}}
\newcommand{\REVrestore}{\textsc{restore}}
\newcommand{\REVfirsturn}{\textsc{firsturn}}
\newcommand{\REVyouturn}{\textsc{youturn}}
\newcommand{\REVcheck}{\textsc{check}}
\newcommand{\REVcapo}{\textsc{capo}}
\newcommand{\REVfine}{\textsc{fine}}
\newcommand{\REVsnaps}{\textsc{snaps}}

\title[Checkpointing for Adjoint Multistage Schemes]{Optimal Checkpointing for Adjoint Multistage Time-Stepping Schemes}
\thanks{This material is based upon work supported by the U.S. Department of Energy, Office of Science, Office of Advanced Scientific Computing Research, Scientific Discovery through Advanced Computing program through the FASTMath Institute and the Base Math under Contract DE-AC02-06CH11357 at Argonne National Laboratory. This research used resources of the National Energy Research Scientific Computing Center (NERSC), a U.S. Department of Energy Office of Science User Facility located at Lawrence Berkeley National Laboratory, operated under Contract No. DE-AC02-05CH11231. 
}

\author{Hong Zhang}
\email{hongzhang@anl.gov}
\author{Emil M. Constantinescu}
\email{emconsta@anl.gov}
\affiliation{%
  \institution{Argonne National Laboratory}
  \city{Lemont}
  \state{IL}
  \country{USA}
  \postcode{60439}
}

\begin{document}

\begin{abstract}
We consider checkpointing strategies that minimize the number of recomputations
needed when performing discrete adjoint computations using multistage
time-stepping schemes that require computing several substeps within one
complete time step. Specifically, we propose two algorithms that can generate
optimal checkpointing schedules under weak assumptions. The first is an
extension of the seminal \texttt{Revolve} algorithm adapted to multistage
schemes. The second algorithm, named \texttt{CAMS}, is developed based on
dynamic programming, and it requires the least number of recomputations when
compared with other algorithms. The \texttt{CAMS} algorithm is made publicly
available in a library with bindings to C and Python. Numerical results show
that the proposed algorithms can deliver up to two times the speedup compared
with that of classical \texttt{Revolve}. Moreover, we discuss the utilization of
the \texttt{CAMS} library in mature scientific computing libraries and
demonstrate the ease of using it in an adjoint workflow. The proposed algorithms
have been adopted by the \texttt{PETSc} \texttt{TSAdjoint} library. Their
performance has been demonstrated with a large-scale PDE-constrained
optimization problem on a leadership-class supercomputer.
\end{abstract}
\keywords{
Checkpointing, adjoint method, multistage time-stepping schemes,  dynamic
programming, PETSc.
}

\ccsdesc[500]{Mathematics of computing~Automatic differentiation, Mathematical software}
\ccsdesc[500]{Computing methodologies~Modeling and simulation, Machine Learning}

\maketitle

\section{Introduction}

Adjoint computation is a key component for solving partial differential equation
(PDE)-constrained optimization, uncertainty quantification problems, and inverse
problems in a wide range of scientific and engineering fields. It has also been
used for training neural networks, where it is called backpropagation in machine
learning. The adjoint method is an efficient way to compute the derivatives of a
scalar-valued function with respect to a large number of independent variables
including system parameters and initial conditions. In this method, the function
can be interpreted as a sequence of operations, and one can compute the
derivatives by applying the chain rule of differentiation to these operations in
reverse order. Therefore, the computational flow of the function is reversed,
starting from the dependent variables (output) and propagating back to the
independent variables (input).
Furthermore, because of the similarity between stepwise function evaluation and
a time-stepping algorithm for solving time-dependent dynamical systems, the
stepwise function evaluation in any programming language can be interpreted as a
time-dependent procedure, where the primitive operation propagates a state
vector from one time step to another. This interpretation makes it easier to
address the challenge of the checkpointing problem described below.

A general nonlinear dynamical system
\begin{equation}
  \dot{\bu} = \f(t,\bu)\,, \quad \bu(t_0) = \bu_0, \quad \bu \in \mathbb{R}^m, \quad t_0 \leq t \leq t_f
  \label{eq:ode_def}
\end{equation}
can be solved by a time-stepping algorithm 
\begin{equation}
\bu_{n+1} = \N (\bu_n), \quad n = 0,\dots,N-1,
\label{eqn:timestepping_operator}
\end{equation} 
where $\N$ is a time-stepping operator that advances the solution from $t_n$ to
$t_{n+1}$. An explicit time-stepping scheme is discussed in Section
\ref{sec:revisiting} as an example of $\N$. To compute the derivative of a
scalar functional $\psi(\bu_N)$ that depends on the final solution $\bu_N$ with
respect to the system initial state $\bu_0$, one can solve the discrete adjoint
equation
\begin{equation}
  \blambda_{n} =  \left(\frac{d \N}{d \bu}(\bu_n)\right)^T \blambda_{n+1}, \quad \blambda \in \mathbb{R}^m, \quad n= N-1, \dots, 0, \\
\label{eqn:disadjoint}
\end{equation}
for the adjoint variable $\blambda$, which is initialized with $\blambda_{N} =
\left(\frac{d \psi}{d \bu}(\bu_N)\right)^T$. Details of the derivation of this
formula can be found in, for example, \cite{zhang2021tsadjoint,Zhang2014}. The
adjoint variable carries the derivative information and propagates it backward
in time from $t_f$ to $t_0$. At each time step in the backward run, the solution
of the original system \eqref{eqn:timestepping_operator} must be available. It
is natural to save the solution vectors when carrying out the forward run and
reuse them in the reverse run. However, the required storage capacity grows
linearly with the problem size $m$ and the number of time steps $N$
\cite{Griewank2008Evaluating}, making this approach intractable for large-scale
long-time simulations. When there is a limited budget for storage, one resorts
to checkpointing the solution at selected time steps and recomputing the missing
solutions needed \cite{Griewank1992}.  This selection complicates the
algorithm implementation, however, and may hamper the overall performance. Therefore, an
optimal checkpointing problem, which will be formally defined later, needs to be
solved in order to minimize the recomputation cost with a budget constraint on
storage capacity.

Two decades ago Griewank and Walther introduced the first offline optimal
checkpointing strategy that addresses this problem
\cite{Griewank2000,Griewank1992}. They established the theoretical relationship
between the number of recomputations and the number of allowable checkpoints
(solution vectors) based on a binomial algorithm. Their algorithm was
implemented in a software library called \texttt{Revolve} and has been used by
automatic differentiation tools such as ADOL-C, ADtool, and Tapenade and by
algorithmic differentiation tools such as dolfin-adjoint \cite{Farrell2013}.
Many follow-up studies have extended this foundational work to tackle online
checkpointing problems, where the number of computation steps is unknown
\cite{Heuveline2006,Stumm2010,Wang2009}, and multistage or multilevel
checkpointing \cite{Aupy2016,Aupy2017,Schanen2016,Herrmann2020} for
heterogeneous storage systems. Note that multistage checkpointing refers to
checkpointing on multiple storage devices (e.g., memory and disk) and is not
related to multistage time integration. In developing these algorithms, one
commonly assumes that memory is limited and the cost of storing/restoring
checkpoints is negligible, as is still the case for many modern computing
architectures.

Despite theoretical advances in checkpointing strategies, implementing a
checkpointing schedule poses a great challenge because of the workflow
complexity, preventing widespread use. The checkpointing schedule requires the
workflow to switch between partial forward sweeps and partial reverse sweeps, a
requirement that is difficult to achieve. For this reason, many adjoint models
implement suboptimal checkpointing schedules
\cite{CVODES,rackauckas2020universal} or even no checkpointing schedules
\cite{Zhang2014,Andersson2019}. In order to mitigate the implementation
difficulty, the classical \texttt{Revolve} library is designed to be a centralized
controller that guides users through the implementation. However, its use
requires a specific intrusive workflow. In contrast, our proposed approach is implemented as an external guide.


Our contributions in this paper are as follows.
\begin{enumerate}
  \item We demonstrate that the classical \texttt{Revolve} algorithm can
  become suboptimal when multistage time-stepping schemes such as Runge--Kutta
  methods are used to solve the dynamical system \eqref{eq:ode_def}.
  \item We present two new algorithms, namely, modified \texttt{Revolve} and
  \texttt{CAMS}, that are suited to general and stiffly accurate multistage
  time-stepping schemes.
  \item We present a library that implements the \texttt{CAMS} algorithm, and we discuss its utilization in adjoint ordinary differential equation (ODE) solvers.
  \item We demonstrate that our library can be easily incorporated into scientific libraries such as \texttt{PETSc} and be used for efficient large-scale adjoint calculation with limited storage budget.
  \item We show the performance advantages of our algorithms over the classical approach with experimental results on a leadership-class supercomputer.
\end{enumerate}

\section{Optimal checkpointing problems in adjoint calculation} \label{sec:revisiting}

\subsection{Number of recomputations as a performance metric}
In a conventional checkpointing strategy, system states at selective time steps
are stored into checkpoints during a forward sweep and restored during a reverse
sweep. Before an adjoint step \eqref{eqn:adjsensi} can be taken, the state
$\bu_n$ is needed to recompute the time step $t_n \rightarrow t_{n+1}$ to
recover all the intermediate information. And if $\bu_n$ is not checkpointed,
one has to recompute from the nearest checkpoint to recover it first. Often
 in the literature (such as \cite{Griewank2000}) authors implicitly assume that the
execution cost of every step is constant and the cost of memory access/copy is
insignificant compared with the calculation cost. Consequently, the number
of recomputations can be used as a good metric that reflects the computational
cost associated with a checkpointing strategy. In the rest of this paper we make
the same assumption and use the same metric when developing our checkpointing
strategies.

\subsection{Conventional checkpointing problem}
The conventional checkpointing problem, defined in Problem \ref{def:problem1} below, has been solved by Griewank and Walther \cite{Griewank2000} with the  \texttt{Revolve} algorithm. \texttt{Revolve} generates a checkpointing schedule based on the assumption that a checkpoint accommodates only the solution at a time step. For one-stage time-stepping schemes such as the Euler method, the solution is guaranteed to be optimal according to Proposition \ref{prop:revolve} proved in \cite{Griewank2000}. 
\begin{problem}[$Prob_{\texttt{conventional}}(m,s)$] Assume the solution at a selective time step can be saved to a checkpoint. Given the number of time steps $m$ and the maximum allowable number of checkpoints $s$, find a checkpointing schedule that minimizes the number of recomputations in the adjoint computation for \eqref{eqn:timestepping_operator}.
  \label{def:problem1}
\end{problem}
\begin{proposition}[Griewank and Walther \cite{Griewank2000}]
  The optimal solution to $Prob_{\texttt{conventional}}(m,s)$ requires a minimal number of recomputations
  \begin{equation}
  p(m,s)= t m-\binom{s+t}{t-1},
  \label{eqn:p}
  \end{equation}
  where
  $t$ is the unique integer (called repetition number in \cite{Griewank2000})
  that satisfies $\binom{s+t-1}{t-1} < m \leq \binom{s+t}{t}$.
  \label{prop:revolve}
\end{proposition}

As an example, we illustrate in Figure \ref{fig:process_a} the optimal reversal
schedule generated by \texttt{Revolve}, which requires the minimal number of
recomputations to reverse $10$ time steps given $3$ maximum allowed checkpoints.
For the convenience of notation, we associate the solution of the ODE system
\eqref{eq:ode_def} at each time step with an index. The initial condition
corresponds to index $0$, and the index increases by $1$ for each successful time
step. In the context of adaptive time integration, a successful time step means
the last time step taken after potentially several attempted steps to determine
a suitable step size. Note that the failed attempts are excluded and not
indexed. Starting from the final time step, the adjoint computation decreases
the index by $1$ after each backward step until the index reaches $0$.

During the forward run, the solutions at indices $0$, $4$, and $7$ are copied
into the first three checkpoints. When the final step $9 \rightarrow 10$ is
finished, the solution and stage values at this step are usually accessible in
memory, so the adjoint computation for the last time step can be taken directly.
To compute the next backward step $9 \rightarrow 8$, one can acquire the
solution at index $9$ and the stage values by restoring the second checkpoint
and recomputing two steps forward in time. The second checkpoint can be
discarded after the backward step $8 \rightarrow 7$ so that its storage can be
reused in following steps. Throughout the entire process, $15$ recomputations are taken, as shown Figure \ref{fig:process_a}. One can verify using Proposition \ref{prop:revolve} that the repetition number for this case is $t=2$ and an optimal solution should have $p(10,3)=15$.
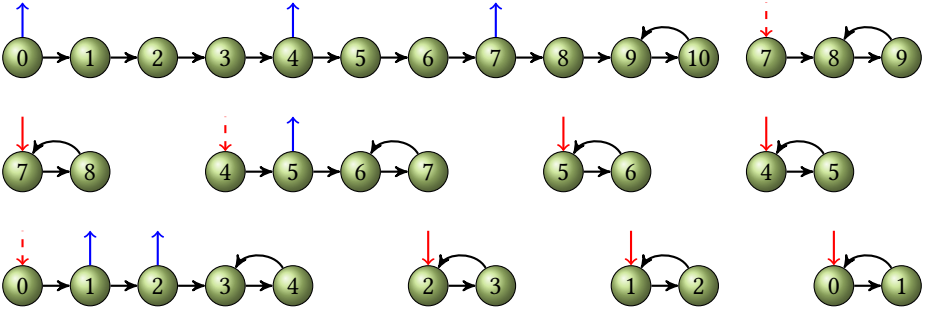
\begin{figure}
  \begin{tikzpicture}[node distance=0.9cm,auto,>=stealth']
    \begin{scope}
      \node [stepnum, pin={[store]}] (sn0) {$0$};
      \node [stepnum] (sn1) [right of=sn0] {$1$}
        edge [edgefrom] (sn0);
      \node [stepnum] (sn2) [right of=sn1] {$2$}
        edge [edgefrom] (sn1);
      \node [stepnum] (sn3) [right of=sn2] {$3$}
        edge [edgefrom] (sn2);
      \node [stepnum, pin={[store]}] (sn4) [right of=sn3] {$4$}
        edge [edgefrom] (sn3);
      \node [stepnum] (sn5) [right of=sn4] {$5$}
        edge [edgefrom] (sn4);
      \node [stepnum] (sn6) [right of=sn5] {$6$}
        edge [edgefrom] (sn5);
      \node [stepnum, pin={[store]}] (sn7) [right of=sn6] {$7$}
        edge [edgefrom] (sn6);
      \node [stepnum] (sn8) [right of=sn7] {$8$}
        edge [edgefrom] (sn7);
      \node [stepnum] (sn9) [right of=sn8] {$9$}
        edge [edgefrom] (sn8);
      \node [stepnum] (sn10) [right of=sn9] {$10$}
        edge [edgefrom] (sn9)
        edge [edgebend] (sn9);
    \end{scope}
    \begin{scope}[shift={(9.9cm,0cm)}]
      \node [stepnum,pin={[restorekeepinstack]}] (sn7) {$7$}; 
      \node [stepnum] (sn8) [right of=sn7] {$8$}
        edge [edgefrom] (sn7);
       \node [stepnum] (sn9) [right of=sn8] {$9$}
        edge [edgefrom] (sn8)  
        edge [edgebend] (sn8);
    \end{scope}
    \begin{scope}[shift={(0cm,-1.5cm)}]
      \node [stepnum,pin={[restore]}] (sn7) {$7$}; 
      \node [stepnum] (sn8) [right of=sn7] {$8$}
        edge [edgefrom] (sn7)
        edge [edgebend] (sn7);
    \end{scope}
    \begin{scope}[shift={(2.7cm,-1.5cm)}]
      \node [stepnum, pin={[restorekeepinstack]}] (sn4) {$4$};
      \node [stepnum,pin={[store]}] (sn5) [right of=sn4] {$5$}
        edge [edgefrom] (sn4);
      \node [stepnum] (sn6) [right of=sn5] {$6$}
        edge [edgefrom] (sn5);
      \node [stepnum] (sn7) [right of=sn6] {$7$}
        edge [edgefrom] (sn6)
        edge [edgebend] (sn6);
    \end{scope}
    \begin{scope}[shift={(7.2cm,-1.5cm)}]
      \node [stepnum,pin={[restore]}] (sn5) {$5$};
      \node [stepnum] (sn6) [right of=sn5] {$6$}
        edge [edgefrom] (sn5)
        edge [edgebend] (sn5);
    \end{scope}
    \begin{scope}[shift={(9.9cm,-1.5cm)}]
      \node [stepnum,pin={[restore]}] (sn4) {$4$};
      \node [stepnum] (sn6) [right of=sn4] {$5$}
        edge [edgefrom] (sn4)
        edge [edgebend] (sn4);
    \end{scope}
    \begin{scope}[shift={(0cm,-3cm)}]
      \node [stepnum, pin={[restorekeepinstack]}] (sn0) {$0$};
      \node [stepnum, pin={[store]}] (sn1) [right of=sn0] {$1$}
        edge [edgefrom] (sn0);
      \node [stepnum, pin={[store]}] (sn2) [right of=sn1] {$2$}
        edge [edgefrom] (sn1);
      \node [stepnum] (sn3) [right of=sn2] {$3$}
        edge [edgefrom] (sn2);
      \node [stepnum] (sn4) [right of=sn3] {$4$}
        edge [edgefrom] (sn3)
        edge [edgebend] (sn3);
    \end{scope}
    \begin{scope}[shift={(5.4cm,-3cm)}]
      \node [stepnum, pin={[restore]}] (sn2) {$2$};
      \node [stepnum] (sn3) [right of=sn2] {$3$}
        edge [edgefrom] (sn2)
        edge [edgebend] (sn2);
    \end{scope}
    \begin{scope}[shift={(8.1cm,-3cm)}]
      \node [stepnum, pin={[restore]}] (sn1) {$1$};
      \node [stepnum] (sn2) [right of=sn1] {$2$}
        edge [edgefrom] (sn1)
        edge [edgebend] (sn1);
    \end{scope}
    \begin{scope}[shift={(10.8cm,-3cm)}]
      \node [stepnum, pin={[restore]}] (sn0) {$0$};
      \node [stepnum] (sn1) [right of=sn0] {$1$}
        edge [edgefrom] (sn0)
        edge [edgebend] (sn0);
    \end{scope}
  \end{tikzpicture}
  \caption{Application of \texttt{Revolve} to reversing $10$ time steps given $3$ allowable
  checkpoints. The up arrow and down arrow stand for the ``store'' and
  ``restore'' operations, respectively. The down arrows with solid lines
  indicate that the checkpointing units can be discarded after being used. The
  down arrows with dashed lines mean to restore the unit without removing it from
  memory. }
  \label{fig:process_a}
\end{figure}

\subsection{Minimizing the number of recomputations for multistage methods}
The conventional strategy, however, can be suboptimal for multistage time-stepping schemes if we relax the assumption to allow saving the intermediate stages together with the solution as a checkpoint. For example, consider the modified problem.
\begin{problem}[$Prob_{\texttt{modified}}(m,s)$] Assume a checkpoint is composed
  of the solution and the stage values at a time step. Given the number of time
  steps $m$ and the maximum allowable number of checkpoints $s$, find a
  checkpointing schedule that minimizes the number of recomputations in the
  adjoint computation for \eqref{eqn:timestepping_operator}.
  \label{def:problem2}
\end{problem}
Multistage schemes such as Runge--Kutta (RK) methods are popular for solving
systems of ODEs; their adjoint counterparts are implemented in ODE solver
libraries such as \texttt{FATODE} \cite{Zhang2014} and more recently by
\texttt{PETSc}
\texttt{TSAdjoint}\cite{PETSc-user-ref,zhang2021tsadjoint,abhyankar2018petscts}.
An $\ell$-stage explicit RK method is expressed as 
\begin{equation}
  \label{eqn:adjsensi}
  \begin{aligned}
  \bU_i & = \bu_n + h_n \sum_{j=1}^{i-1}  \, a_{ij} \, \f(\bU_j), \quad i=1,\cdots,\ell, \\
  \bu_{n+1} & = \bu_n + h_n \sum_{i=1}^\ell  \, b_i \, \f(\bU_i).
  \end{aligned}
\end{equation}
Its discrete adjoint is
\begin{equation}
\label{eqn:disadj_erk}
\begin{aligned}
  \blambda_{\ell,i} & =  h_n \fu^T(\textcolor{blue}{\bU_i})  \left( b_i \blambda_{n+1} + \sum_{j=i+1}^\ell a_{ji} \, \blambda_{\ell,j} \right), \quad i=\ell,\cdots,1\\
  \blambda_n & = \blambda_{n+1} + \sum_{j=1}^\ell \blambda_{\ell,j},
\end{aligned}
\end{equation}
where $\blambda$ is the adjoint variable that carries the sensitivity
information and is propagated in a backward step during a reverse sweep.

The adjoint step of a Runge--Kutta scheme requires all the stage values (see the
sensitivity equation \eqref{eqn:disadj_erk}, for example). These intermediate values are
usually obtained by restoring a state saved during the forward sweep and
recomputing a time step using this state, as illustrated in Figure
\ref{fig:rk_checkpoint}. \texttt{Revolve} requires its users to implement a
basic action (named $youturn$) that recomputes a forward step followed
immediately by a backward step. This strategy has been a de facto standard in
classical adjoint computation.
\begin{figure}[ht]
  \centering
  \includegraphics[width=0.9\textwidth]{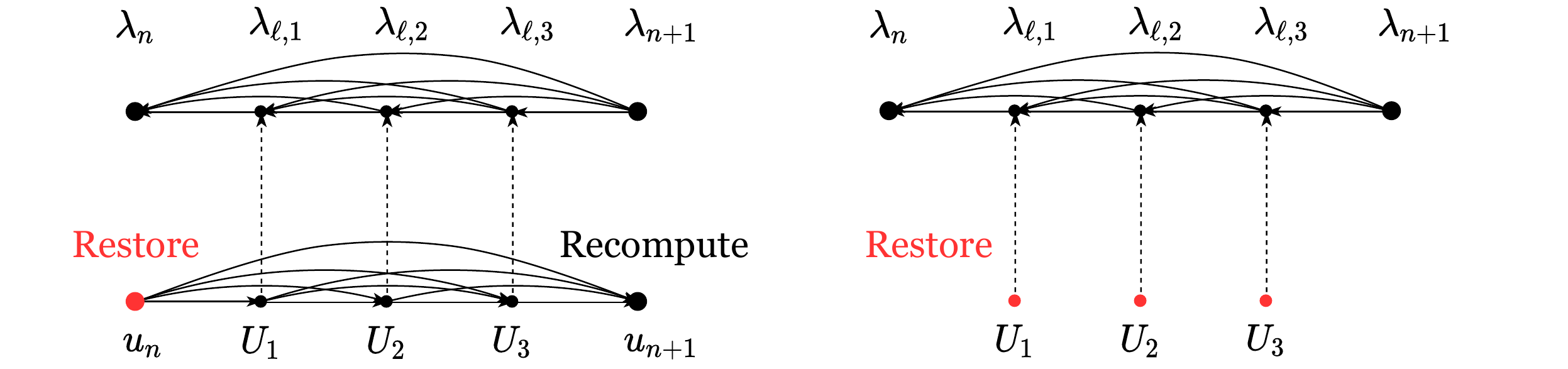}
  \caption{Strategies for obtaining the stage values during the reverse sweep of a discrete adjoint solver based on a Runge--Kutta method. Left is the classical strategy used with \texttt{Revolve}. Right is the strategy considered in this work.}
  \label{fig:rk_checkpoint}
\end{figure}

Let us count the number of recomputations required for an ideal case where the
memory is unlimited. To reverse $m$ time steps, $m-1$ recomputations would be
required if one checkpointed the solution at every time step. No
recomputation would be needed,  however, if one checkpointed the stage values instead of
the solution for all the time steps. 

For cases where memory is limited, checkpointing the stage values may still lead
to fewer recomputations. Based on this observation, we extend the classical
optimal checkpointing scheme in \cite{Griewank2000} to solve Problem \ref{def:problem2}. See Figure \ref{fig:rk_checkpoint}
for a schematic illustration. Although saving more information at each time step
means that fewer checkpoints are available, we will show in Section \ref{sec:mrevolve} that in certain circumstances
the extended scheme may still outperform the original scheme, with the gain
depending on the total number of time steps to be reversed.

Our goal is to minimize the number of recomputations under practical assumptions. Further relaxing the assumption, we will seek in Section \ref{sec:acms} an optimal solution for the following problem that requires the fewest recomputations among the solutions to the three problems ($\textsc{Problem} 1,2,3$).
\begin{problem}[$Prob_{\texttt{multistage}}(m,s)$] Assume that either the solution  or an intermediate stage (which has the same size as the solution) can be saved to a checkpoint. Given the number of time steps $m$ and the maximum allowable number of checkpoints $s$, find a checkpointing schedule that minimizes the number of recomputations in the adjoint computation for \eqref{eqn:timestepping_operator}.
  \label{def:problem3}
\end{problem}

\section{Modified checkpointing scheme based on \texttt{Revolve}}\label{sec:mrevolve} 

In this section we first describe the \texttt{Revolve} nomenclature in order to
provide the necessary background. We then  introduce our solution to
$Prob_{\texttt{modified}}(m,s)$ using a concrete example, discuss its optimality,
and address the implementation aspects.

\subsection{Modification to the \texttt{Revolve} offline algorithm}

To generate the schedule in Figure \ref{fig:process_a}, one needs to call the
API routine \textsc{revolve}() repeatedly and implement the actions prescribed
by the algorithm. The return value of \textsc{revolve}() tells the calling
program to perform one of the actions among \REVadvance{},
\REVtakeshot{}(\REVstore{}), \REVrestore{}, \REVfirsturn{}, and \REVyouturn{},
which are briefly summarized in Table \ref{tab:revolve} and are explained in
detail in \cite{Griewank2000}.
\begin{table}[ht]
  \caption{\texttt{Revolve} nomenclature.}
  \label{tab:revolve}
  \centering
  \begin{tabular}{l l}
  \toprule
      \REVadvance{} & advance the solution forward \\
      \REVtakeshot{}(\REVstore{})  &  copy the solution into a checkpoint \\
      \REVrestore{} & copy a checkpoint back into the solution \\
      \REVfirsturn{} & take one backward step directly (usually after one forward
      step) \\
      \REVyouturn{} & take one forward step and then one backward step \\
   \bottomrule
  \end{tabular}
\end{table}

In our modification of the schedule, every checkpoint position is shifted by
one, and the stage values are included so that the Jacobian can be computed
directly from these. To be specific, if the original \texttt{Revolve} algorithm
determines that one should checkpoint the solution at index $i$, we will store a
combined checkpoint at the end of the time step $i \rightarrow i+1$ including
the solution at index $i+1$ and the stage values. We do so by mapping the
actions prescribed by \texttt{Revolve} to a series of new but similar actions,
while guaranteeing the optimality for the new checkpointing settings. Table
\ref{tab:mrevolve} lists the mappings we conduct in the modified schedule.
Figure \ref{fig:process_b} illustrates the checkpointing schedule generated by
the modified \texttt{Revolve} algorithm.
\begin{table}[ht]
  \centering
  \caption{Mapping the \texttt{Revolve} output to new actions.}
  \label{tab:mrevolve}
  \begin{tabular}{l l}
  \toprule
  \texttt{Revolve}   &  Our modification \\
  \midrule
      \REVadvance{} from $i$ to $j$ & advance the solution from $i+1$ to $j+1$\\
      \REVstore{} solution $i$ &  copy the solution at $i+1$ and the stages into
      a checkpoint \\
      \REVrestore{} to solution $i$ & restore the solution at $i+1$ and the stages from a checkpoint \\
      \REVyouturn{} & take one backward step directly (\REVfirsturn{}) \\
   \bottomrule
  \end{tabular}
\end{table}
\begin{figure}
  \begin{tikzpicture}[node distance=0.9cm,auto,>=stealth']
    \begin{scope}
      \node [stepnum] (sn0) {$0$};
      \node [stepnum, pin={[store]}, label={[label distance=0.2cm,stagevalue,pin={[store]}]above left:}] (sn1) [right of=sn0] {$1$}
        edge [edgefrom] (sn0);
      \node [stepnum] (sn2) [right of=sn1] {$2$}
        edge [edgefrom] (sn1);
      \node [stepnum] (sn3) [right of=sn2] {$3$}
        edge [edgefrom] (sn2);
      \node [stepnum] (sn4) [right of=sn3] {$4$}
        edge [edgefrom] (sn3);
      \node [stepnum, pin={[store]}, label={[label distance=0.2cm,stagevalue,pin={[store]}]above left:}] (sn5) [right of=sn4] {$5$}
        edge [edgefrom] (sn4);
      \node [stepnum] (sn6) [right of=sn5] {$6$}
        edge [edgefrom] (sn5);
      \node [stepnum] (sn7) [right of=sn6] {$7$}
        edge [edgefrom] (sn6);
      \node [stepnum, pin={[store]}, label={[label distance=0.2cm,stagevalue,pin={[store]}]above left:}] (sn8) [right of=sn7] {$8$}
        edge [edgefrom] (sn7);
      \node [stepnum] (sn9) [right of=sn8] {$9$}
        edge [edgefrom] (sn8);
      \node [stepnum] (sn10) [right of=sn9] {$10$}
        edge [edgefrom] (sn9)
        edge [edgebend] (sn9);
    \end{scope}
    \begin{scope}[shift={(10.8cm,0cm)}]
      \node [stepnum,pin={[restorekeepinstack]}] (sn8) {$8$};
      \node [stepnum] (sn9) [right of=sn8] {$9$}
        edge [edgefrom] (sn8)  
        edge [edgebend] (sn8);
    \end{scope}
    \begin{scope}[shift={(0cm,-1.5cm)}]
      \node [stepnum] (sn7) {$7$}; 
      \node [stepnum,pin={[restore]}, label={[label distance=0.2cm,stagevalue,pin={[restore]}]above left:}] (sn8) [right of=sn7] {$8$}
        edge [edgebend] (sn7);
    \end{scope}
    \begin{scope}[shift={(3.6cm,-1.5cm)}]
      \node [stepnum, pin={[restorekeepinstack]}] (sn5) {$5$};
      \node [stepnum,pin={[store]},label={[label distance=0.2cm,stagevalue,pin={[store]}]above left:}] (sn6) [right of=sn5] {$6$}
        edge [edgefrom] (sn5);
      \node [stepnum] (sn7) [right of=sn6] {$7$}
        edge [edgefrom] (sn6)
        edge [edgebend] (sn6);
    \end{scope}
    \begin{scope}[shift={(7.2cm,-1.5cm)}]
      \node [stepnum] (sn5) {$5$};
      \node [stepnum,pin={[restore]},label={[label distance=0.2cm,stagevalue,pin={[restore]}]above left:}] (sn6) [right of=sn5] {$6$}
        edge [edgebend] (sn5);
    \end{scope}
    \begin{scope}[shift={(9.9cm,-1.5cm)}]
      \node [stepnum] (sn4) {$4$};
      \node [stepnum,pin={[restore]},label={[label distance=0.2cm,stagevalue,pin={[restore]}]above left:}] (sn6) [right of=sn4] {$5$}
        edge [edgebend] (sn4);
    \end{scope}
    \begin{scope}[shift={(0.9cm,-3cm)}]
      \node [stepnum, pin={[restorekeepinstack]}] (sn1) {$1$};
      \node [stepnum,pin={[store]},label={[label distance=0.2cm,stagevalue,pin={[store]}]above left:}] (sn2) [right of=sn1] {$2$}
        edge [edgefrom] (sn1);
      \node [stepnum,pin={[store]},label={[label distance=0.2cm,stagevalue,pin={[store]}]above left:}] (sn3) [right of=sn2] {$3$}
        edge [edgefrom] (sn2);
      \node [stepnum] (sn4) [right of=sn3] {$4$}
        edge [edgefrom] (sn3)
        edge [edgebend] (sn3);
    \end{scope}
    \begin{scope}[shift={(5.4cm,-3cm)}]
      \node [stepnum] (sn2) {$2$};
      \node [stepnum,pin={[restore]},label={[label distance=0.2cm,stagevalue,pin={[restore]}]above left:}] (sn3) [right of=sn2] {$3$}
        edge [edgebend] (sn2);
    \end{scope}
    \begin{scope}[shift={(8.1cm,-3cm)}]
      \node [stepnum] (sn1) {$1$};
      \node [stepnum,pin={[restore]},label={[label distance=0.2cm,stagevalue,pin={[restore]}]above left:}] (sn2) [right of=sn1] {$2$}
        edge [edgebend] (sn1);
    \end{scope}
    \begin{scope}[shift={(10.8cm,-3cm)}]
      \node [stepnum] (sn0) {$0$};
      \node [stepnum,pin={[restore]},label={[label distance=0.2cm,stagevalue,pin={[restore]}]above left:}] (sn1) [right of=sn0] {$1$}
        edge [edgebend] (sn0);
    \end{scope}
  \end{tikzpicture}
  \caption{Application of the modified
  \texttt{Revolve} to reversing $10$ time steps given $3$ allowable
  checkpoints. The up arrow and down arrow stand for the ``store'' and
  ``restore'' operations, respectively. The down arrows with solid lines
  indicate that the checkpointing units can be discarded after being used. The
  down arrows with dashed lines mean to restore the unit without removing it from
  memory. }
  \label{fig:process_b}
\end{figure}
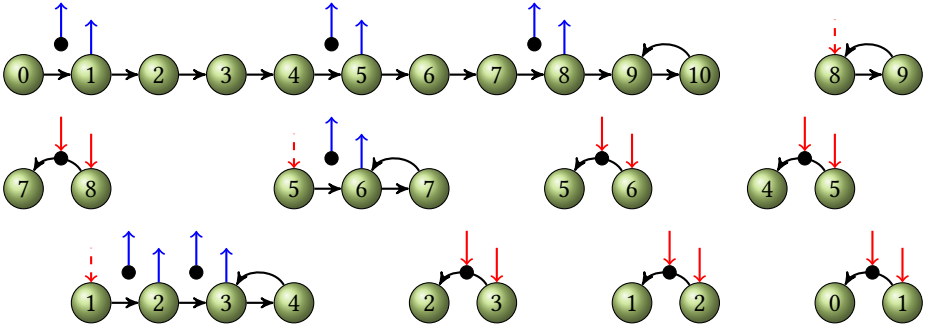

\subsection{Optimality of the modified algorithm}

As we observe from the optimal schedule for adjoint computation provided by
\texttt{Revolve}, the intermediate information required by each backward step
must be recomputed from the solution restored from a checkpoint, except for the
last time step. Because every checkpoint is shifted to one time step later, the
distance between the restored checkpoint and the current solution in the reverse
sweep is reduced by one,  thus saving exactly one recomputation for each adjoint
step. This observation leads to the following proposition regarding optimality
for the modified \texttt{Revolve} algorithm.
\begin{proposition}
Assume that a checkpoint is composed of the solution and the stage values at a time
step. Given $s$ number of allowed checkpoints in memory, the minimal number of
additional forward steps (recomputations) needed for the adjoint computation of
$m$ time steps is 
\begin{equation}
\tilde{p}(m,s)= (t-1)\, m-\binom{s+t}{t-1}+1,
\label{eqn:ptilde}
\end{equation}
where  
$t$ is the unique integer that satisfies $\binom{s+t-1}{t-1} < m \leq \binom{s+t}{t}$.
\label{prop:recomp}
\end{proposition}
\begin{proof}
According to the
observation mentioned above, one can further save $m-1$ additional forward steps
with the modified scheme. We will prove by contradiction that no further savings
are possible with other schedules.

If  a schedule exists that satisfies the assumption and takes fewer
recomputations than \eqref{eqn:ptilde}, one can move all the checkpoints
backward by one step and exclude the stage values so that the checkpoints are
composed of solutions only. The resulting schedule will cost $m-1$ additional
recomputations by construction. It is clearly a solution to the classical
checkpointing problem for $m$ time steps given $s$ allowed checkpoints, and the
total number of recomputations required is less than $ t\, m-\binom{s+t}{t-1}$.
This contradicts the optimality result in Proposition \ref{prop:revolve}.
\end{proof}

\begin{remark}
The same modification can also be applied to the online checkpointing algorithms
in \cite{Stumm2010,Heuveline2006,Wang2009} and to the multistage checkpointing
algorithms in \cite{Stumm2009}. Considering the same number of allowable
checkpoints, the number of recomputations saved is always equal to the total
number of steps minus one. An immediate question is whether to use the modified
or the original algorithm (in other words, whether to save stage values) given
the same amount of storage capacity. 

According to Proposition \ref{prop:recomp}, the optimal option depends on
multiple factors, including the number of steps, the number of stages of the
time-stepping algorithm, and the memory capacity. Therefore, for a given
time-stepping algorithm and a fixed amount of storage capacity, the best choice
can be made based on only the total number of time steps. For example, we
suppose there is storage space for $12$ solutions when using the original
algorithm. In the modified algorithm, the same space can be used to store $6$
checkpoints if each checkpoint consists of one solution and one stage, and it
can be used to store $4$ checkpoints if each checkpoint consists of one solution
and two stages. In Figure \ref{fig:mrerolve} we plot the number of
recomputations for the two algorithms in two scenarios---saving one additional
stage and saving two additional stages. As can be seen, saving the stages
together with the solution is more favorable than saving only the solution until
a crossover point is reached ($41$ and $13$ steps, respectively, for the two
illustrated scenarios). Furthermore, the number of stages, determined by the
time-stepping scheme, has a critical impact on the location of the crossover
point; for schemes with fewer stages, saving the stages benefits a wider range
of time steps (compare $41$ with $13$).
\begin{figure}
  \centering
  \includegraphics[width=0.95\textwidth]{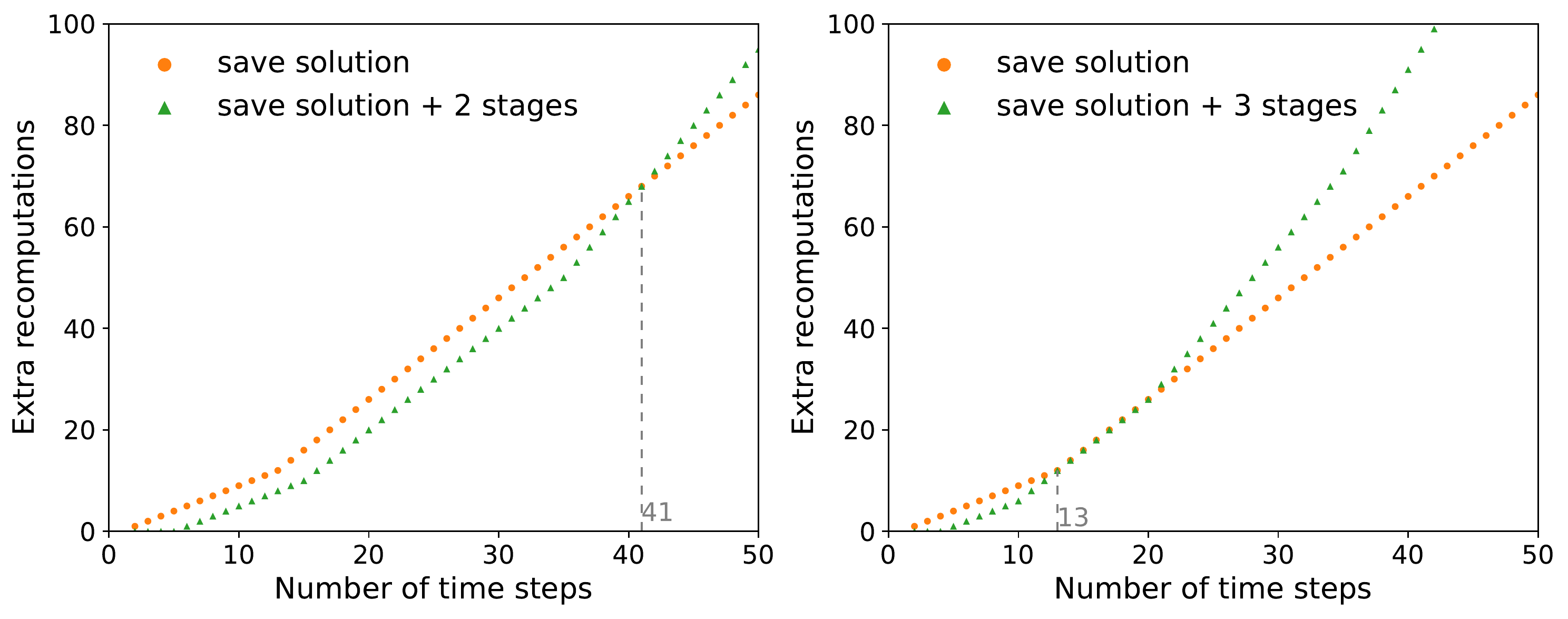}
  \caption{Minimal recomputations required by \texttt{Revolve} and the modified
  \texttt{Revolve} with different numbers of stages ($2$ in the left plot and
  $3$ in the right plot). 
 Adapted from \cite{zhang2021tsadjoint}.}
  \label{fig:mrerolve}
\end{figure}
\end{remark}
\begin{remark}
Some multistage time-stepping methods have a property that can be
exploited to reduce the size of a checkpoint. In particular, if the last stage is
equal to the solution at the end of a time step, one can skip the last stage
when storing a checkpoint. Many classical implicit RK methods and the
Crank--Nicolson method have this property and thus can be treated as having one
less stage. 
\end{remark}

\section{Truly optimal checkpointing for the adjoint of multistage schemes}\label{sec:acms}

Both the \texttt{Revolve} algorithm and the modified \texttt{Revolve} we propose
are proved to be optimal under some assumptions on the checkpointing strategy
(see Figure \ref{fig:rk_checkpoint}); however, neither of them is ideal in
practice when evaluated without these assumptions, as indicated in Figure
\ref{fig:mrerolve}. To obtain a truly optimal checkpointing schedule, we
consider the solution to $Prob_{\texttt{multistage}}(m,s)$ with a further
relaxed assumption allowing the solution or the stage values or both at one time
step to be checkpointed, a situation that is not difficult to achieve in most
ODE solvers. For the fairness of comparison, we introduce the concept of a
checkpointing unit. By definition, one checkpointing unit can store one solution
vector or one stage vector since they have the same size; one checkpoint may
contain one or more units, thus having different types. The total memory (in
bytes) occupied by checkpoints is $number\ of\ checkpointing\ units \times
checkpoining\ unit\ size$ (in bytes). Therefore, when comparing different
algorithms, using the same number of checkpointing units indicates using the
same amount of total memory.

In addition, many multistage schemes are constructed to be stiffly accurate in
order to solve stiff ODEs. This property typically requires that the solution at the
end of each time step be equal to the last stage of the method. Taking into
account this observation as well as the relaxed assumption, we develop the
optimal checkpointing algorithm \texttt{CAMS} that includes two variants, one
for stiffly accurate schemes and the other for general cases, denoted by
\texttt{CAMS-SA} and \texttt{CAMS-GEN}, respectively. Both variants are
developed by using a divide-and-conquer strategy.

\subsection{\texttt{CAMS} for stiffly accurate multistage schemes} \label{sec:stifflyaccurate}

First, let us consider a subproblem of the checkpointing problem
$Prob_{\texttt{multistage}}(m,s)$. If the initial state is already checkpointed,
we want to know how many additional forward steps are necessary to reverse a
sequence of time steps with a given amount of checkpointing units. To solve this subproblem, we establish a recurrence equation as follows.
\begin{lemma}
Given $s$ allowed checkpointing units in memory and the \textbf{initial state}
stored in memory, the minimal number of additional forward steps needed for the
adjoint computation of $m$ time steps using an $\ell$-stage time integrator
satisfies
\begin{equation}
    P_\textnormal{IS}(m,s) = \min
    \begin{cases}
    \min_{1 \leq \tilde{m} \leq m} \left(\tilde{m} + P_\textnormal{IS}(\tilde{m},s) + P_\textnormal{IS}(m-\tilde{m},s-1)\right)\\
    \min_{1 \leq \tilde{m} \leq m} \left(\tilde{m}-1 + P_\textnormal{IS}(\tilde{m}-1,s) + P_\textnormal{IS}(m-\tilde{m},s-\ell)\right)
    \end{cases} .
    \label{eq:p1}
\end{equation}
\label{th:p1}
\end{lemma}

\begin{proof}
Assume that the next checkpoint is taken after $\tilde{m}$ steps with $1 \leq
\tilde{m} \leq m$. We need to consider two cases based on the type of the next
checkpoint.

\textit{Case 1:} The state $\tilde{m}$ is checkpointed.

\textit{Case 2:} The state $\tilde{m}$ and stage values at time step $\tilde{m}$
are checkpointed.

For case 1, the sequence of $m$ time steps can be split into two parts, one with
$\tilde{m}$ time steps and the other with $m-\tilde{m}$. The second part will be
reversed first, requiring $P_\textnormal{IS}(m-\tilde{m},s-1)$ additional
forward steps. The first part needs a forward sweep over the $m$ time steps
before the reverse run can be performed. Summing up all the additional forward
steps leads to the first recurrence equation in \ref{eq:p1}.

For case 2, the $m$th time step can be reversed directly since the stage values
can be restored from memory; therefore, the first subsequence consists of $m-1$
time steps. The second subsequence still consists of $m-\tilde{m}$ time steps,
but there are $\ell$ fewer checkpointing units available. This case corresponds
to the second recurrence equation in \ref{eq:p1}.
\end{proof}

With the solution of the reduced problem, we can solve the original problem
easily. Notice that to reverse the first time step, one can either restore the
initial state from checkpoints and recompute the time step or restore the stage
values directly from memory. The latter option transforms the problem into
checkpointing for $m-1$ time steps given $s-\ell$ checkpointing units and the
first state already checkpointed, which is $P_\textnormal{IS}(m-1,s-\ell+1)$.
\begin{theorem}
Given $s$ allowed checkpointing units in memory, the minimal number of
additional forward steps needed for the adjoint computation of $m$ time steps
using an $\ell$-stage time integrator is
\begin{equation}
    P(m,s) = \min \left( P_\textnormal{IS}(m,s), P_\textnormal{IS}(m-1,s-\ell+1) \right) .
    \label{eq:p}
\end{equation}
\end{theorem}

Based on Lemma \ref{th:p1}, we can design a dynamic programming algorithm to
compute $P_\textnormal{IS}(m,s)$ and save the results in a table. Computing
$P(m,s)$ requires querying only two values from the table. Given $m,s,\ell$ and
the index $lastcptstep$ and the type $lastcpttype$ of the last checkpoint, the
routine \texttt{CAMS}($lastcptstep,lastcpttype,s,m,\ell$) returns the index
$nextcptstep$ and the type $nextcptype$ of the next checkpoint. These two
returned variables are determined internally according to the choice of minimum
made in \eqref{eq:p1}. In particular, $nextcptstep$ depends on $\tilde{m}$, and
$nextcptype$ depends on which case of the two yields the minimum. The index $nextcptstep$
gives the position of the next checkpoint. The type $nextcptype$ can be either solution
or stage values. This information, often dubbed the \textbf{path} to reach the
minimum, can also be stored in tables when computing  $P_\textnormal{IS}(m,s)$
recursively. Tabulation of intermediate results is a standard procedure in
dynamic programming; thus we do not detail it here.

\subsection{\texttt{CAMS} for general multistage schemes} \label{sec:gmm}

Depending on how the first checkpoint is created, we split the problem into two
scenarios: (1) the initial state is checkpointed, and (2) the stage values of
the first step are checkpointed. The corresponding recurrence equations are
established in Lemmas \ref{th:pis_normal} and \ref{th:psv_normal}, respectively,
and are used to generate the final result in Theorem \ref{th:p_normal}. The
proofs are similar to the results in Section \ref{sec:stifflyaccurate}. Note
that the two subproblems are intertwined in the derivation, and they are solved
with \textbf{double dynamic programming}.
\begin{lemma}
Given $s$ allowed checkpointing units in memory and the \textbf{initial state}
stored in memory, the minimal number of additional forward steps needed for the
adjoint computation of $m$ time steps using an $\ell$-stage time integrator
satisfies 
\begin{equation}
    P_\textnormal{IS}(m,s) = \min
    \begin{cases}
     \min_{1 \leq \tilde{m} \leq m-2} \left(\tilde{m} + P_\textnormal{IS}(\tilde{m},s) + P_\textnormal{IS}(m-\tilde{m},s-1) \right)\\
    \min_{2 \leq \tilde{m} \leq m-1} \left(\tilde{m}-1 + P_\textnormal{IS}(\tilde{m}-1,s) + P_\textnormal{SV}(m-\tilde{m}+1,s-1) \right)
    \end{cases} .
    \label{eq:pis_normal}
\end{equation}
\label{th:pis_normal}
\end{lemma}

\begin{lemma}
Given $s$ allowed checkpointing units in memory and the \textbf{stage values} of
the first time step stored in memory, the minimal number of additional forward
steps needed for the adjoint computation of $m$ time steps using an $\ell$-stage
time integrator satisfies 
\begin{equation}
    P_\textnormal{SV}(m,s) = \min \left( P_\textnormal{IS}(m-1,s-\ell), P_\textnormal{SV}(m-1,s-\ell) \right) .
    \label{eq:psv_normal}
\end{equation}
\label{th:psv_normal}
\end{lemma}

\begin{theorem}
Given $s$ allowed checkpointing units in memory, the minimal number of
additional forward steps needed for the adjoint computation of $m$ time steps
using an $\ell$-stage time integrator is
\begin{equation}
    P(m,s) = \min \left( P_\textnormal{IS}(m,s), P_\textnormal{SV}(m,s) \right) \,.
    \label{eq:p_normal}
\end{equation}
\label{th:p_normal}
\end{theorem}

The correctness of Lemmas \ref{th:pis_normal} and \ref{th:psv_normal} is
illustrated in Figure \ref{fig:p1p2}. For $P_\textnormal{IS}(m,s)$, we suppose
that the second checkpoint is placed after $\tilde{m}$th time step during the
forward sweep. Depending on the type of the second checkpoint, the second
subproblem, which reverses the sequence of time steps starting from the second
checkpoint, can be addressed by $P_\textnormal{IS}(m-\tilde{m},s-1)$ and
$P_\textnormal{SV}(m-\tilde{m-1}+1,s-1)$, respectively. During the reverse
sweep, the second subproblem is solved first, and then the initial state is
restored from the first checkpoint and integrated to the location of the second
checkpoint. Thus, $\tilde{m}$ or $\tilde{m}-1$ additional time steps are needed
for solving the first subproblem. A similar strategy can be applied for
$P_\textnormal{SV}(m,s)$ as well. But note that the second checkpoint must be
placed after the second time step; otherwise, the second time step cannot be
reversed.
\begin{figure}[ht]
  \centering
  \subfloat[$P_\textnormal{IS}(m,s)$]{
    \includegraphics[width=0.4\textwidth]{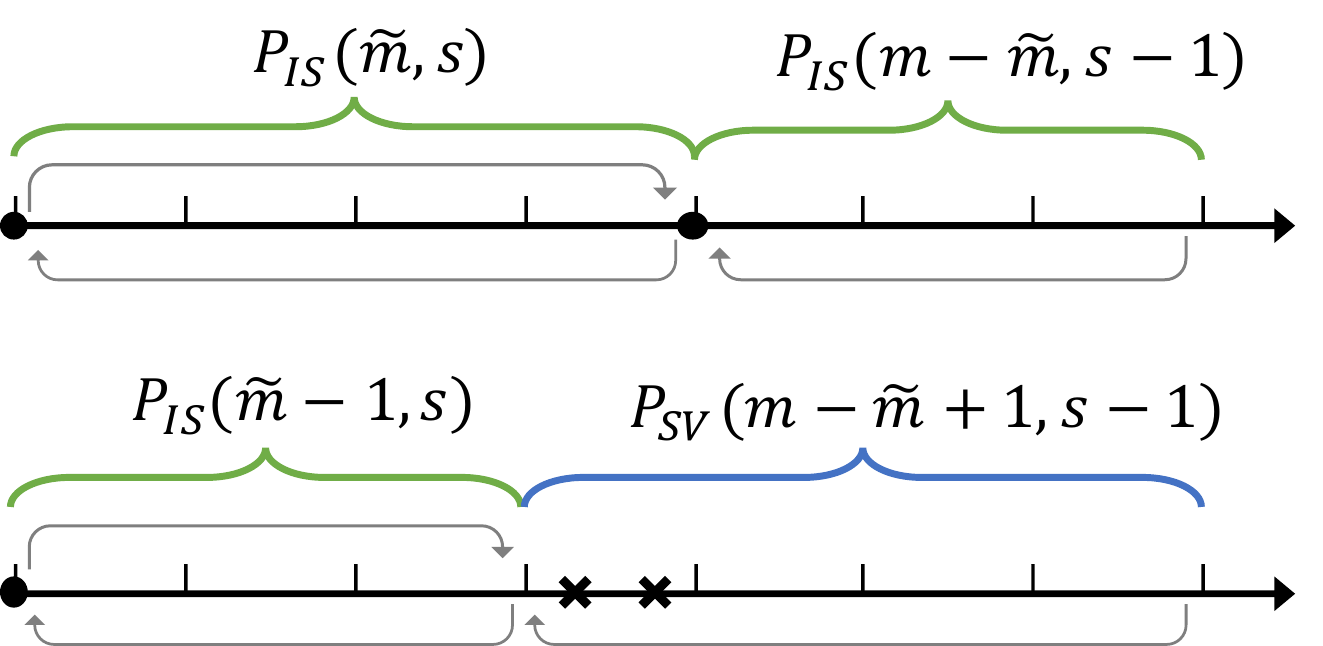}
  }
  \hfill
  \subfloat[$P_\textnormal{SV}(m,s)$]{
    \includegraphics[width=0.4\textwidth]{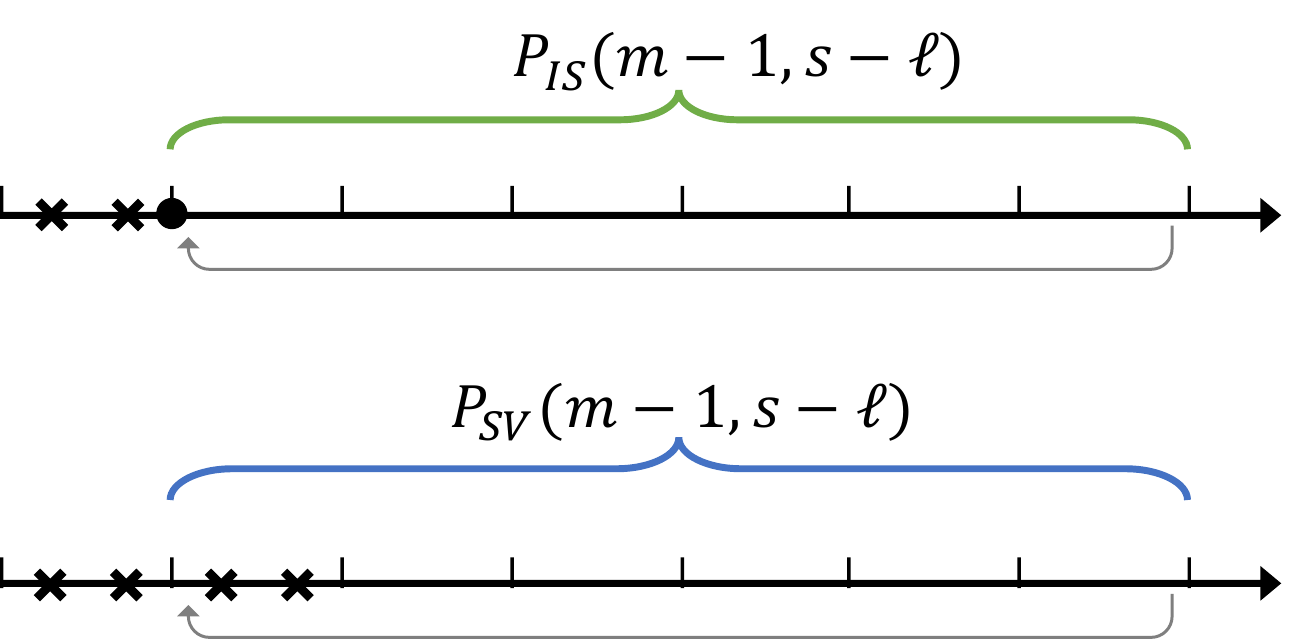}
  }
  \caption{Strategy to solve the intertwined subproblems
  $P_\textnormal{IS}(m,s)$ and $P_\textnormal{SV}(m,s)$. Each subproblem is
  broken into smaller subproblems. A dot stands for a solution, and crosses
  stand for stage values; $m,s,\ell$ are the number of time steps, the number of
  checkpointing units, and the number of stages, respectively.}
  \label{fig:p1p2}
\end{figure}

Theorem \ref{th:p_normal} combines the solution of the two subproblems, and the
proof is straightforward. Based on Lemmas \ref{th:pis_normal} and
\ref{th:psv_normal}, we can develop dynamic programming algorithms to compute
$P_\textnormal{IS}(i,j)$ and $P_\textnormal{SV}(i,j)$, respectively, and
tabulate the values and the path information for any input $i \leq m$ and $j\leq
s$. The resulting algorithm for the adjoint computation is almost the same as
Algorithm \ref{alg:cams} and thus is omitted for brevity---except that an
additional type of the checkpoint is considered. Specifically, the choices of a
checkpoint type include solution only, stage values only, and solution plus
stage values. Note that Lemmas \ref{th:pis_normal} and \ref{th:psv_normal}
distinguish between the first two choices. If the next checkpoint is stage
values and the checkpoint after the next is a solution calculated directly from
the stage values, it makes the implementation easier to fuse the two checkpoints
into a new type of checkpoint because both of them are available at the end of a
time step.

\subsection{Performance analysis for \text{CAMS}}

\begin{proposition}
Using the \texttt{CAMS} algorithm takes no more recomputations than using the
\texttt{Revolve} algorithm and the modified \texttt{Revolve} algorithm.
\end{proposition}

\begin{proof}
By construction, to reverse $m$ time steps with $s$ checkpointing units, the
\texttt{Revolve} algorithm requires the number of recomputations to satisfy
\begin{equation}
    P_\textnormal{R}(m,s) = \min_{1\leq \tilde{m} \leq m} \left( \tilde{m} + P_\textnormal{R}(\tilde{m},s) + P_\textnormal{R}(m-\tilde{m},s-1) \right),
\end{equation}
as shown in \cite{Griewank2000}.

Following the same methodology, we can derive a recurrence equation for the
modified \texttt{Revolve} algorithm: 
\begin{equation}
    P_\textnormal{MR}(m,s) = \min_{1\leq \tilde{m} \leq m} \left( \tilde{m}-1 + P_\textnormal{MR}(\tilde{m}-1,s) + P_\textnormal{MR}(m-\tilde{m}+1,s-1-\ell) \right).
\end{equation}

Now we compare $P_\textnormal{R}(m,s)$ and $P_\textnormal{MR}(m,s)$ with $P(m,s)$ obtained with our
\texttt{CAMS} algorithm. According to \eqref{eq:p_normal}, $P(m,s)$ is always
less than or equal to $P_\textnormal{IS}(m,s)$. The first case in
\eqref{eq:pis_normal} indicates $P_\textnormal{IS}(m,s) \leq P_\textnormal{R}(m,s)$. The
second case in \eqref{eq:pis_normal} indicates
\begin{align*}
P_\textnormal{IS}(m,s) \leq      \min_{2 \leq \tilde{m} \leq m-1} \left(\tilde{m}-1 + P_\textnormal{IS}(\tilde{m}-1,s) + P_\textnormal{SV}(m-\tilde{m}+1,s-1) \right) \\
    \leq \min_{2 \leq \tilde{m} \leq m-1} \left(\tilde{m}-1 + P_\textnormal{IS}(\tilde{m}-1,s) + P_\textnormal{IS}(m-\tilde{m},s-1-\ell) \right) .
\end{align*}
Note that $P_\textnormal{IS}(i,j)$ is a convex function in $i \in {1,\dots,m}$.
Therefore,
\[
P_\textnormal{IS}(m-\tilde{m},s-1-\ell) \leq P_\textnormal{IS}(m-\tilde{m}+1,s-1-\ell).
\]
Then we obtain
\[
P_\textnormal{IS}(m,s) \leq P_\textnormal{MR}(m,s).
\]
\end{proof}

We can conclude with this proposition that the number of recomputations needed
by \texttt{CAMS} is bounded by the number of recomputations needed by
\texttt{Revolve}, which is given in \eqref{eqn:p}. We note that the lower bound
for \texttt{Revolve} is $m-1$ whereas \texttt{CAMS} takes zero recomputations if
there is sufficient memory.

Figure \ref{fig:acms_10_6} illustrates the application of \texttt{Revolve} and
\texttt{CAMS} for reversing $10$ time steps given $6$ allowable checkpointing
units. Each application starts from a forward time integration, at the end of
which all the checkpointing slots are filled. The distribution of the
checkpoints is determined by the corresponding checkpointing algorithm (e.g.,
Alg. \ref{alg:cams}). In the backward integration, a solution can be restored
from checkpoints and then used for recomputing the intermediate states; for
\texttt{CAMS}, stage values can be restored and used for reversing the time step
directly without any recomputation. As a result, \texttt{revolve} takes $12$
recomputations whereas \texttt{CAMS} takes $8$ recomputations, or $6$
recomputations if the integration method is stiffly accurate.
\begin{figure}[!ht]
\centering
\subfloat[\texttt{Revolve} (12 recomputations)]{
\begin{tikzpicture}[node distance=0.9cm,auto,>=stealth']
  \begin{scope}
    \node [stepnum, pin={[store]}] (sn0) {$0$};
    \node [stepnum, pin={[store]}] (sn1) [right of=sn0] {$1$}
      edge [edgefrom] (sn0);
    \node [stepnum, pin={[store]}] (sn2) [right of=sn1] {$2$}
      edge [edgefrom] (sn1);
    \node [stepnum, pin={[store]}] (sn3) [right of=sn2] {$3$}
      edge [edgefrom] (sn2);
    \node [stepnum] (sn4) [right of=sn3] {$4$}
      edge [edgefrom] (sn3);
    \node [stepnum, pin={[store]}] (sn5) [right of=sn4] {$5$}
      edge [edgefrom] (sn4);
    \node [stepnum] (sn6) [right of=sn5] {$6$}
      edge [edgefrom] (sn5);
    \node [stepnum, pin={[store]}] (sn7) [right of=sn6] {$7$}
      edge [edgefrom] (sn6);
    \node [stepnum] (sn8) [right of=sn7] {$8$}
      edge [edgefrom] (sn7);
    \node [stepnum] (sn9) [right of=sn8] {$9$}
      edge [edgefrom] (sn8);
    \node [stepnum] (sn10) [right of=sn9] {$10$}
      edge [edgefrom] (sn9)
      edge [edgebend] (sn9);
  \end{scope}
  \begin{scope}[shift={(9.9cm,0cm)}]
    \node [stepnum,pin={[restorekeepinstack]}] (sn7) {$7$}; 
    \node [stepnum] (sn8) [right of=sn7] {$8$}
      edge [edgefrom] (sn7);
     \node [stepnum] (sn9) [right of=sn8] {$9$}
      edge [edgefrom] (sn8)  
      edge [edgebend] (sn8);
  \end{scope}
  \begin{scope}[shift={(0cm,-1.5cm)}]
    \node [stepnum,pin={[restore]}] (sn7) {$7$}; 
    \node [stepnum] (sn8) [right of=sn7] {$8$}
      edge [edgefrom] (sn7)
      edge [edgebend] (sn7);
  \end{scope}
  \begin{scope}[shift={(2.7cm,-1.5cm)}]
    \node [stepnum,pin={[restorekeepinstack]}] (sn5) {$5$};
    \node [stepnum] (sn6) [right of=sn5] {$6$}
      edge [edgefrom] (sn5);
    \node [stepnum] (sn7) [right of=sn6] {$7$}
      edge [edgefrom] (sn6)
      edge [edgebend] (sn6);
  \end{scope}
  \begin{scope}[shift={(6.3cm,-1.5cm)}]
    \node [stepnum,pin={[restore]}] (sn5) {$5$};
    \node [stepnum] (sn6) [right of=sn5] {$6$}
      edge [edgefrom] (sn5)
      edge [edgebend] (sn5);
  \end{scope}
  \begin{scope}[shift={(9.9cm,-1.5cm)}]
  \node [stepnum,pin={[restore]}] (sn3) {$3$};
    \node [stepnum] (sn4) [right of=sn3] {$4$}
      edge [edgefrom] (sn3);
    \node [stepnum] (sn6) [right of=sn4] {$5$}
      edge [edgefrom] (sn4)
      edge [edgebend] (sn4);
  \end{scope}
  \begin{scope}[shift={(0cm,-3cm)}]
    \node [stepnum, pin={[restorekeepinstack]}] (sn3) {$3$};
    \node [stepnum] (sn4) [right of=sn3] {$4$}
      edge [edgefrom] (sn3)
      edge [edgebend] (sn3);
  \end{scope}
  \begin{scope}[shift={(2.7cm,-3cm)}]
    \node [stepnum, pin={[restore]}] (sn2) {$2$};
    \node [stepnum] (sn3) [right of=sn2] {$3$}
      edge [edgefrom] (sn2)
      edge [edgebend] (sn2);
  \end{scope}
  \begin{scope}[shift={(5.4cm,-3cm)}]
    \node [stepnum, pin={[restore]}] (sn1) {$1$};
    \node [stepnum] (sn2) [right of=sn1] {$2$}
      edge [edgefrom] (sn1)
      edge [edgebend] (sn1);
  \end{scope}
  \begin{scope}[shift={(8.1cm,-3cm)}]
    \node [stepnum, pin={[restore]}] (sn0) {$0$};
    \node [stepnum] (sn1) [right of=sn0] {$1$}
      edge [edgefrom] (sn0)
      edge [edgebend] (sn0);
  \end{scope}
\end{tikzpicture}
\label{fig:revolve}
}\\
\subfloat[\texttt{CAMS} with stiffly accurate two-stage schemes (6 recomputations)]{
\begin{tikzpicture}[node distance=0.9cm,auto,>=stealth']
  \begin{scope}
    \node [stepnum] (sn0) {$0$};
    \node [stepnum, label={[label,yshift=0.35cm,stagevalue,pin={[store]}] left :},label={[label distance=0.2cm,yshift=0.35cm,stagevalue,pin={[store]}] left:}] (sn1) [right of=sn0] {$1$}
      edge [edgefrom] (sn0);
    \node [stepnum] (sn2) [right of=sn1] {$2$}
      edge [edgefrom] (sn1);
    \node [stepnum] (sn3) [right of=sn2] {$3$}
      edge [edgefrom] (sn2);
    \node [stepnum] (sn4) [right of=sn3] {$4$}
      edge [edgefrom] (sn3);
    \node [stepnum, label={[label,yshift=0.35cm,stagevalue,pin={[store]}] left :},label={[label distance=0.2cm,yshift=0.35cm,stagevalue,pin={[store]}] left:}] (sn5) [right of=sn4] {$5$}
      edge [edgefrom] (sn4);
    \node [stepnum] (sn6) [right of=sn5] {$6$}
      edge [edgefrom] (sn5);
    \node [stepnum] (sn7) [right of=sn6] {$7$}
      edge [edgefrom] (sn6);
    \node [stepnum, label={[label,distance=0.05cm,yshift=0.3cm,stagevalue,pin={[store]}] left :},label={[label distance=0.25cm,yshift=0.3cm,stagevalue,pin={[store]}] left:}] (sn8) [right of=sn7] {$8$}
      edge [edgefrom] (sn7);
    \node [stepnum] (sn9) [right of=sn8] {$9$}
      edge [edgefrom] (sn8);
    \node [stepnum] (sn10) [right of=sn9] {$10$}
      edge [edgefrom] (sn9)
      edge [edgebend] (sn9);
  \end{scope}
  \begin{scope}[shift={(10.8cm,0cm)}]
    \node [stepnum,pin={[restorekeepinstack]}] (sn8) {$8$};
    \node [stepnum] (sn9) [right of=sn8] {$9$}
      edge [edgefrom] (sn8)  
      edge [edgebend] (sn8);
  \end{scope}
  \begin{scope}[shift={(0cm,-1.5cm)}]
    \node [stepnum] (sn7) {$7$}; 
    \node [stepnum, label={[label,yshift=0.35cm,stagevalue,pin={[restore]}]left:}, label={[label distance=0.2cm,yshift=0.35cm,stagevalue,pin={[restore]}]left:}] (sn8) [right of=sn7] {$8$}
      edge [edgebend] (sn7);
  \end{scope}
  \begin{scope}[shift={(3.6cm,-1.5cm)}]
    \node [stepnum, pin={[restorekeepinstack]}] (sn5) {$5$};
    \node [stepnum,label={[label,yshift=0.35cm,stagevalue,pin={[store]}] left :},label={[label distance=0.2cm,yshift=0.35cm,stagevalue,pin={[store]}] left:}] (sn6) [right of=sn5] {$6$}
      edge [edgefrom] (sn5);
    \node [stepnum] (sn7) [right of=sn6] {$7$}
      edge [edgefrom] (sn6)
      edge [edgebend] (sn6);
  \end{scope}
  \begin{scope}[shift={(7.2cm,-1.5cm)}]
    \node [stepnum] (sn5) {$5$};
    \node [stepnum,label={[label,yshift=0.35cm,stagevalue,pin={[restore]}]left:}, label={[label distance=0.2cm,yshift=0.35cm,stagevalue,pin={[restore]}]left:}] (sn6) [right of=sn5] {$6$}
      edge [edgebend] (sn5);
  \end{scope}
  \begin{scope}[shift={(10.8cm,-1.5cm)}]
    \node [stepnum] (sn4) {$4$};
    \node [stepnum,label={[label,yshift=0.35cm,stagevalue,pin={[restore]}]left:}, label={[label distance=0.2cm,yshift=0.35cm,stagevalue,pin={[restore]}]left:}] (sn6) [right of=sn4] {$5$}
      edge [edgebend] (sn4);
  \end{scope}
  \begin{scope}[shift={(0cm,-3cm)}]
    \node [stepnum, pin={[restorekeepinstack]}] (sn1) {$1$};
    \node [stepnum,label={[label,yshift=0.35cm,stagevalue,pin={[store]}] left :},label={[label distance=0.2cm,yshift=0.35cm,stagevalue,pin={[store]}] left:}] (sn2) [right of=sn1] {$2$} edge [edgefrom] (sn1);
    \node [stepnum,label={[label,yshift=0.35cm,stagevalue,pin={[store]}] left :},label={[label distance=0.2cm,yshift=0.35cm,stagevalue,pin={[store]}] left:}] (sn3) [right of=sn2] {$3$}
      edge [edgefrom] (sn2);
    \node [stepnum] (sn4) [right of=sn3] {$4$}
      edge [edgefrom] (sn3)
      edge [edgebend] (sn3);
  \end{scope}
  \begin{scope}[shift={(5.4cm,-3cm)}]
    \node [stepnum] (sn2) {$2$};
    \node [stepnum,label={[label, yshift=0.35cm,stagevalue,pin={[restore]}]left:}, label={[label distance=0.2cm,yshift=0.35cm,stagevalue,pin={[restore]}]left:}] (sn3) [right of=sn2] {$3$}
      edge [edgebend] (sn2);
  \end{scope}
  \begin{scope}[shift={(8.1cm,-3cm)}]
    \node [stepnum] (sn1) {$1$};
    \node [stepnum,label={[label, yshift=0.35cm,stagevalue,pin={[restore]}]left:}, label={[label distance=0.2cm,yshift=0.35cm,stagevalue,pin={[restore]}]left:}] (sn2) [right of=sn1] {$2$}
      edge [edgebend] (sn1);
  \end{scope}
  \begin{scope}[shift={(10.8cm,-3cm)}]
    \node [stepnum] (sn0) {$0$};
    \node [stepnum,label={[label, yshift=0.35cm,stagevalue,pin={[restore]}]left:}, label={[label distance=0.2cm,yshift=0.35cm,stagevalue,pin={[restore]}]left:}] (sn1) [right of=sn0] {$1$}
      edge [edgebend] (sn0);
  \end{scope}
\end{tikzpicture}
\label{fig:stifflyaccurate}
}\\
\subfloat[\texttt{CAMS} with general two-stage schemes ($8$ recomputations)]{
\begin{tikzpicture}[node distance=0.9cm,auto,>=stealth']
  \begin{scope}
    \node [stepnum, pin={[store]}] (sn0) {$0$};
    \node [stepnum] (sn1) [right of=sn0] {$1$}
      edge [edgefrom] (sn0);
    \node [stepnum, pin={[store]}] (sn2) [right of=sn1] {$2$}
      edge [edgefrom] (sn1);
    \node [stepnum] (sn3) [right of=sn2] {$3$}
      edge [edgefrom] (sn2);
    \node [stepnum, pin={[store]}] (sn4) [right of=sn3] {$4$}
      edge [edgefrom] (sn3);
    \node [stepnum] (sn5) [right of=sn4] {$5$}
      edge [edgefrom] (sn4);
    \node [stepnum, pin={[store]}] (sn6) [right of=sn5] {$6$}
      edge [edgefrom] (sn5);
    \node [stepnum] (sn7) [right of=sn6] {$7$}
      edge [edgefrom] (sn6);
    \node [stepnum] (sn8) [right of=sn7] {$8$}
      edge [edgefrom] (sn7);
    \node [stepnum,label={[label,yshift=0.35cm,stagevalue,pin={[store]}] left :},label={[label distance=0.2cm,yshift=0.35cm,stagevalue,pin={[store]}] left:}] (sn9) [right of=sn8] {$9$}
      edge [edgefrom] (sn8);
    \node [stepnum] (sn10) [right of=sn9] {$10$}
      edge [edgefrom] (sn9)
      edge [edgebend] (sn9);
  \end{scope}
  \begin{scope}[shift={(9.9cm,0cm)}]
    \node [stepnum] (sn8) {$8$};
     \node [stepnum,label={[label, yshift=0.35cm,stagevalue,pin={[restore]}]left:}, label={[label distance=0.2cm,yshift=0.35cm,stagevalue,pin={[restore]}]left:}] (sn9) [right of=sn8] {$9$}
      edge [edgebend] (sn8);
  \end{scope}
  \begin{scope}[shift={(0cm,-1.5cm)}]
    \node [stepnum, pin={[restorekeepinstack]}] (sn6) {$6$}; 
    \node [stepnum,label={[label,yshift=0.35cm,stagevalue,pin={[store]}] left :},label={[label distance=0.2cm,yshift=0.35cm,stagevalue,pin={[store]}] left:}] (sn7) [right of=sn6] {$7$}
    edge [edgefrom] (sn6);
    \node [stepnum] (sn8) [right of=sn7] {$8$}
      edge [edgefrom] (sn7)
      edge [edgebend] (sn7);
  \end{scope}
  \begin{scope}[shift={(2.7cm,-1.5cm)}]
    \node [stepnum] (sn6) {$6$};
    \node [stepnum,label={[label, yshift=0.35cm,stagevalue,pin={[restore]}]left:}, label={[label distance=0.2cm,yshift=0.35cm,stagevalue,pin={[restore]}]left:}] (sn7) [right of=sn6] {$7$}
      edge [edgebend] (sn6);
  \end{scope}
  \begin{scope}[shift={(7.2cm,-1.5cm)}]
    \node [stepnum, pin={[restorekeepinstack]}] (sn4) {$4$}; 
    \node [stepnum,label={[label,yshift=0.35cm,stagevalue,pin={[store]}] left :},label={[label distance=0.2cm,yshift=0.35cm,stagevalue,pin={[store]}] left:}] (sn5) [right of=sn4] {$5$}
      edge [edgefrom] (sn4);
    \node [stepnum] (sn6) [right of=sn5] {$6$}
      edge [edgefrom] (sn5)
      edge [edgebend] (sn5);
  \end{scope}
  \begin{scope}[shift={(9.9cm,-1.5cm)}]
    \node [stepnum] (sn4) {$4$};
    \node [stepnum,label={[label, yshift=0.35cm,stagevalue,pin={[restore]}]left:}, label={[label distance=0.2cm,yshift=0.35cm,stagevalue,pin={[restore]}]left:}] (sn5) [right of=sn4] {$5$}
      edge [edgebend] (sn4);
  \end{scope}
  \begin{scope}[shift={(0cm,-3cm)}]
    \node [stepnum, pin={[restore]}] (sn2) {$2$}; 
    \node [stepnum,label={[label,yshift=0.35cm,stagevalue,pin={[store]}] left :},label={[label distance=0.2cm,yshift=0.35cm,stagevalue,pin={[store]}] left:}] (sn3) [right of=sn2] {$3$}
      edge [edgefrom] (sn2);
    \node [stepnum] (sn4) [right of=sn3] {$4$}
      edge [edgefrom] (sn3)
      edge [edgebend] (sn3);
  \end{scope}
  \begin{scope}[shift={(2.7cm,-3cm)}]
    \node [stepnum] (sn2) {$2$};
    \node [stepnum,label={[label, yshift=0.35cm,stagevalue,pin={[restore]}]left:}, label={[label distance=0.2cm,yshift=0.35cm,stagevalue,pin={[restore]}]left:}] (sn3) [right of=sn2] {$3$}
      edge [edgebend] (sn2);
  \end{scope}
  \begin{scope}[shift={(7.2cm,-3cm)}]
    \node [stepnum, pin={[restore]}] (sn0) {$0$}; 
    \node [stepnum,label={[label,yshift=0.35cm,stagevalue,pin={[store]}] left :},label={[label distance=0.2cm,yshift=0.35cm,stagevalue,pin={[store]}] left:}] (sn1) [right of=sn0] {$1$}
      edge [edgefrom] (sn0);
    \node [stepnum] (sn2) [right of=sn1] {$2$}
      edge [edgefrom] (sn1)
      edge [edgebend] (sn1);
  \end{scope}
  \begin{scope}[shift={(9.9cm,-3cm)}]
    \node [stepnum] (sn0) {$0$};
    \node [stepnum,label={[label, yshift=0.35cm,stagevalue,pin={[restore]}]left:}, label={[label distance=0.2cm,yshift=0.35cm,stagevalue,pin={[restore]}]left:}] (sn1) [right of=sn0] {$1$}
      edge [edgebend] (sn0);
  \end{scope}
\end{tikzpicture}
\label{fig:normal}
}
\caption{Application of \texttt{CAMS} and \texttt{Revolve} to reversing $10$
time steps given $6$ allowable checkpointing units. The up arrow and down arrow
stand for the``store'' operation and ``restore'' operation, respectively. The down arrows with solid lines indicate that the checkpointing units can be discarded after being used. The down arrows with dashed lines mean to restore the unit without removing it from memory.}
\label{fig:acms_10_6}
\end{figure}
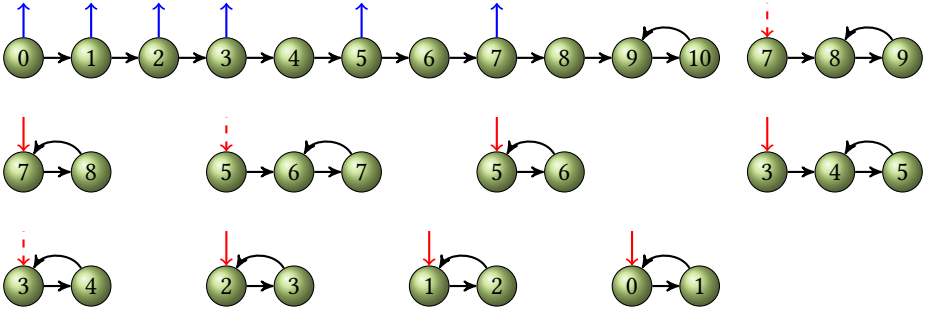
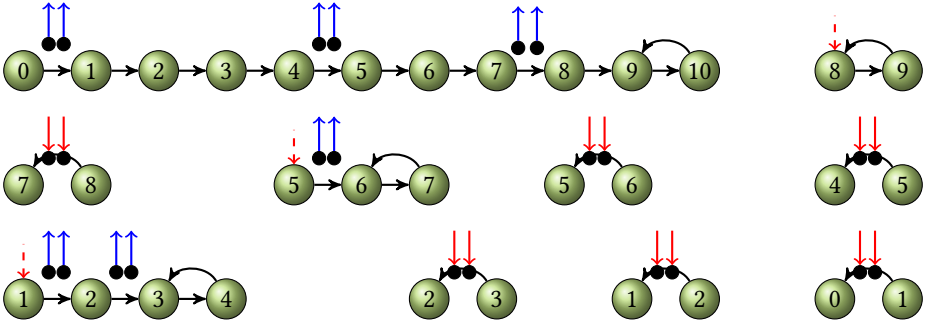
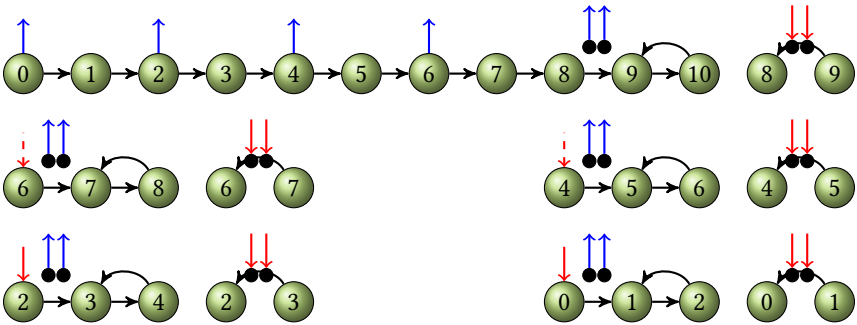

To further study the performance of our algorithms, we plot in Figure
\ref{fig:acms} the actual number of recomputations taken in the algorithm versus
the number of time steps to be reversed, and we compare our algorithms with the
classical \texttt{Revolve} algorithm. For a fair comparison, the same number of
checkpointing units is considered so that the total amount of memory usage is
the same for both algorithms. Figure \ref{fig:acms} shows that the two
\texttt{CAMS} variants outperform \texttt{Revolve} significantly. For $30$
checkpointing units and $300$ time steps, \texttt{CAMS-GEN} takes $210$ fewer
recomputations than \texttt{Revolve} does, and \texttt{CAMS-SA} takes $211$
fewer recomputations than \texttt{Revolve} does. If we assume the computational
cost of a forward step is constant, then the result implies an approximate
speedup of $1.6$ times in running time for the adjoint computation. For $60$
checkpointing units and $300$ time steps, \texttt{CAMS-GEN} and \texttt{CAMS-SA}
result in $261$ and $269$ fewer recomputations, respectively, which can be
translated into an estimated speedup of $2$ times. As the number of time steps
increases, the gap between \texttt{Revolve} and \texttt{CAMS} can be further
enlarged. The modified \texttt{Revolve} also performs better than
\texttt{Revolve} but will eventually be surpassed by \texttt{Revolve}.
Furthermore, when the number of time steps is small, which usually means there
is sufficient memory, no recomputation is needed by modified \texttt{Revolve},
whereas \texttt{Revolve} requires the number of recomputations to be at least as
large as one less than the number of time steps.
\begin{figure}[ht]
  \centering
  \includegraphics[width=0.95\textwidth]{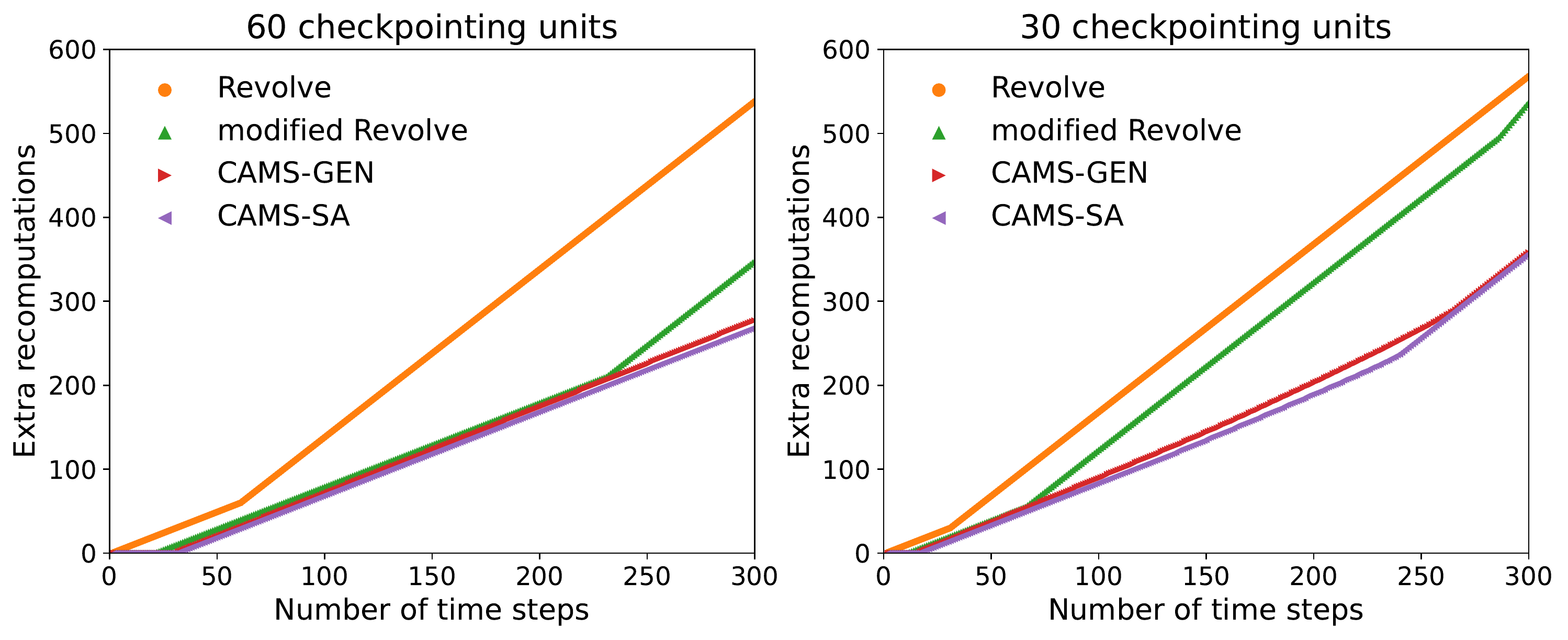}
  \caption{ 
  Performance comparison of various checkpointing algorithms. The plotted data is generated for $2$-stage time-stepping schemes ($\ell=2$).}
  \label{fig:acms}
\end{figure}

Figure \ref{fig:acms2} shows how the needed recomputations vary with the number
of allowable checkpointing units for a fixed number of time steps. This
indicates the memory requirement of each algorithm for a particular
time-to-solution budget.
In the ideal scenario when there is sufficient memory, using modified
\texttt{Revolve} or \texttt{CAMS} can avoid recomputation completely, making
the reverse sweep two times faster than using \texttt{Revolve}, provided the
cost of a forward step is comparable to the cost of a backward step.
\begin{figure}[ht]
  \centering
  \includegraphics[width=0.95\textwidth]{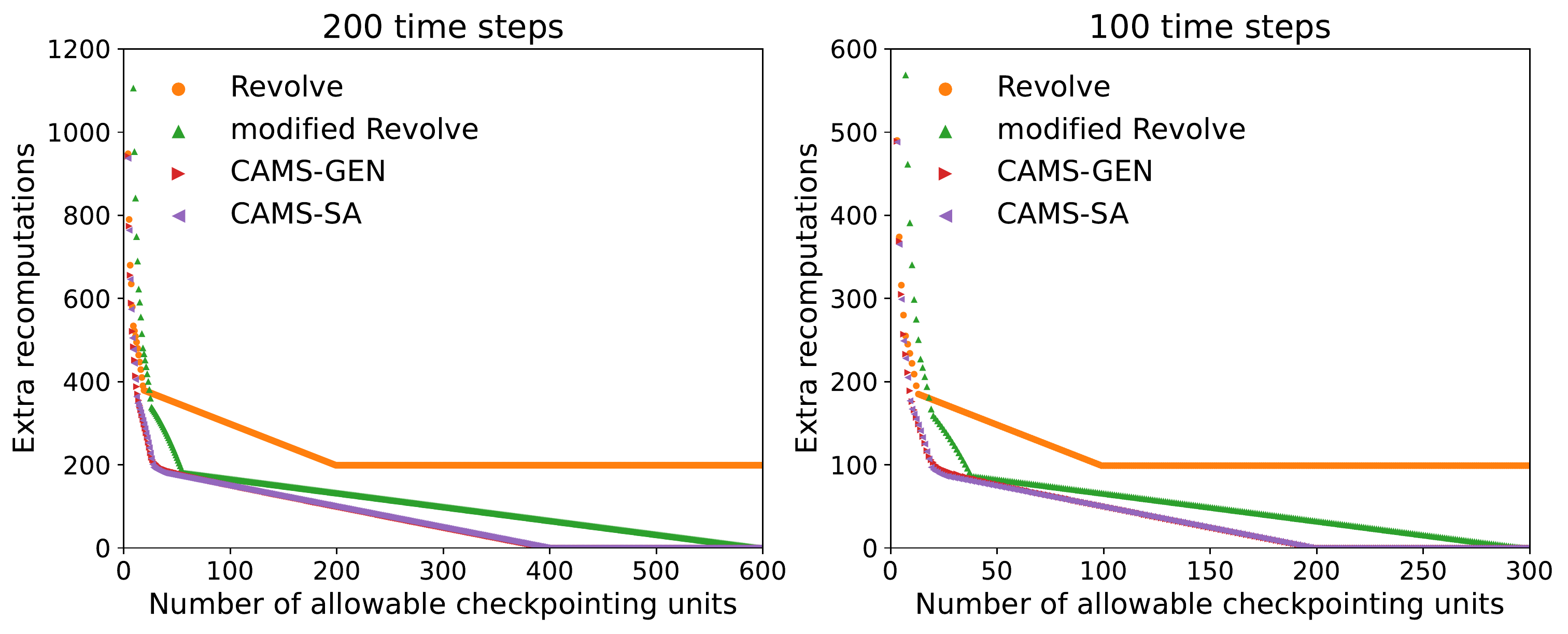}
  \caption{Performance comparison of various checkpointing algorithms. The plotted data is computed for time integration methods with two stages ($\ell=2$).}
  \label{fig:acms2}
\end{figure}

\section{Utilization of \text{Revolve} and \text{CAMS} in a discrete adjoint ODE solver}

\subsection{Using \texttt{Revolve} in a discrete adjoint ODE solver}

The \texttt{Revolve} library is designed to provide explicit control for
conducting forward sweeps and reverse sweeps in adjoint computations. Its user
must implement primitive operations such as performing a forward and backward
step, saving/restoring a checkpoint, and executing these operations in the order
guided by \texttt{Revolve}. Thus, it can be intrusive to incorporate
\texttt{Revolve} in other simulation software such as \texttt{PETSc}; and the
workflow can be difficult to manage, especially when the software has an
established framework for time integration and adaptive time-step control. To
mitigate the intrusive effects, we use \texttt{Revolve} differently so that its
role becomes that of a ``\textbf{consultant}'' rather than a ``\textbf{controller}.'' Algorithm
\ref{alg:revolve} describes our workflow for the adjoint computation with
checkpointing. \texttt{Revolve} relies on several key parameters---\REVcapo{},
\REVfine{}, and \REVcheck{}---and updates them internally. These parameters are
described in Table \ref{tab:revolve_parameter}. For  ease of implementation,
we use a counter \texttt{stepstogo} to track the number of steps to advance. The
routine \texttt{forwardSweep(ind,n,state)} advances the solution $n$ steps from
the $ind$th time point, which is easy to implement in ODE solvers. It can be
used to perform a full forward run and can also be reused to perform the
recomputations in the backward adjoint run. \texttt{revolveForward}() wraps calls
to the API routine \textsc{revolve}() and is intended to guide the selection of
checkpoints. Because each call to \textsc{revolve}() changes the internal states
of \textsc{revolve}, we must carefully control the number of times
\textsc{revolve}() is called in \texttt{revolveForward}() based on the counter
\texttt{stepstogo} and the returned value of the previous calls to
\textsc{revolve}().

Despite the convoluted manipulations of calls to \textsc{revolve}(),  the resulting checkpointing schedule is equivalent to the
original one generated by calling \textsc{revolve}() repeatedly (the
``controller'' mode). This is justified primarily by two observations.
\begin{itemize}
\item When \textsc{revolve}() returns \REVtakeshot{} (which means storing a
checkpoint), the next call to \textsc{revolve}() will return either \REVadvance{} or \REVyouturn{} or \REVfirsturn{}.
\item In the reverse sweep, every backward step is preceded by restoring a
checkpoint and recomputing from this point.
\end{itemize}
\begin{table}[ht]
  \centering
  \caption{Key \texttt{Revolve} parameters}
  \label{tab:revolve_parameter}
  \begin{tabular}{l l}
  \toprule
      \REVcapo{}  & the starting index of the time-step range to be reversed \\
      \REVfine{}  & the ending index of the time-step range to be reversed \\
      \REVcheck{} & the number of checkpoints in use \\
      \REVsnaps{} & the maximum number of checkpoints allowed\\
   \bottomrule
  \end{tabular}
\end{table}
\begin{algorithm}
  \caption{Proposed adjoint computation for a sequence of $m$ time steps using \texttt{Revolve}.}
\begin{algorithmic}[1]
  \State Initialize global variables $capo \coloneqq 0$, $fine \coloneqq m$, $check \coloneqq -1$, $snaps \coloneqq s$
  \State Initialize global variable $stepstogo \coloneqq 0$ 
  \State Initialize \texttt{Revolve} with $capo,fine,check$, and $snaps$ 
  \State $state \leftarrow$ \Call{forwardSweep}{$0,m,state$} 
  \State $adjstate \leftarrow$ \Call{adjointStep}{$adjstate$} 
  \For{$i \coloneqq m-1$ to $1$} 
    \State $whatodo \leftarrow $ \Call{revolve}{check,capo,fine,snaps}
    \State \Assert{whatodo=restore} \Comment{always start from restoring a checkpoint}
    \State $state,restoredind \leftarrow$ \Call{restore}{$check$} 
    \State $state \leftarrow $ \Call{forwardSweep}{$restoredind,i-restoredind,state$} 
    \State $adjstate \leftarrow$ \Call{adjointStep}{$adjstate$}
  \EndFor
  \Function{{forwardSweep}}{$ind,n,state$} \Comment{advance n steps from the \textit{ind}-th point}
  \For{$i \leftarrow ind$ to $ind+n-1$} 
    \State \Call{revolveForward} {$state$} \Comment{return \texttt{youturn}/\texttt{firsturn} at last iteration}
    \State $state \leftarrow$ \Call{forwardStep}{$state$} 
  \EndFor 
  \State \Return state 
  \EndFunction
  \Function{{revolveForward}}{$state$} \Comment{query \texttt{Revolve} and take actions}
    \If{$stepstogo=0$}
      \State $oldcapo \coloneqq capo$
      \State $whatodo \leftarrow $ \Call{revolve}{check,capo,fine,snaps}
      \If{$whatodo = takeshot$}  
        \State \Call{store}{$state,check$}  
          \State $oldcapo \coloneqq capo$
          \State $whatodo \leftarrow $ \Call{revolve}{check,capo,fine,snaps}
      \EndIf
      \If{$whatodo = restore$}  
        \State \Call{restore}{$state,check$}  
        \State $oldcapo \coloneqq capo$
        \State $whatodo \leftarrow $ \Call{revolve}{check,capo,fine,snaps}
      \EndIf
      \State \Assert{whatodo=advance||whatodo=youturn||whatodo=firsturn}
      \State $stepstogo \coloneqq capo-oldcapo$
    \Else
      \State $stepstogo \coloneqq stepstogo-1$
    \EndIf
  \EndFunction
\end{algorithmic}
\label{alg:revolve}
\end{algorithm}
Algorithm \ref{alg:mrevolve} depicts the adjoint computation using the modified
\texttt{Revolve}. Compared with Algorithm \ref{alg:revolve}, Algorithm
\ref{alg:mrevolve} shifts the positions of all the checkpoints so that the call
to \textsc{revolve}() is lagged.  We note that decreasing the counter
\texttt{stepstogo} by one in the backward time loop (Line 22) indicates that one
less recomputation is needed for each backward step. 
\begin{algorithm}
  \caption{Proposed adjoint computation using the modified \texttt{Revolve} algorithm.}
\begin{algorithmic}[1]
  \State Initialize global variables $capo \coloneqq 0$, $fine \coloneqq m$, $check \coloneqq -1$, $snaps \coloneqq s$
  \State Initialize global variable $stepstogo \coloneqq 0$
  \State Initialize \texttt{Revolve} with $capo,fine,check$, and $snaps$
  \State $state \leftarrow$ \Call{forwardSweep}{$0,m,state$}
  \State $adjstate \leftarrow$ \Call{adjointStep}{$adjstate$}
  \For{$i \coloneqq M-1$ to $1$}
    \State $restoredind \leftarrow$ \Call{revolveBackward}{$state$} \Comment{always restore a checkpoint}
    \State $state \leftarrow $ \Call{forwardSweep}{$restoredind,i-restoredind,state$}
    \State $adjstate \leftarrow$ \Call{adjointStep}{$adjstate$}
  \EndFor
  \Function{{forwardSweep}}{$ind,n,state$}
  \For{$i \coloneqq ind$ to $ind+n-1$}
    \State $state \leftarrow$ \Call{forwardStep}{$state$}
    \State \Call{revolveForward} {$state$}
  \EndFor
  \State \Return $state$
  \EndFunction
  \Function{{revolveBackward}}{$state$}
    \State $whatodo \leftarrow $ \Call{revolve}{check,capo,fine,snaps}
    \State \Assert{$whatodo = restore$}
    \State $state,restoredind \leftarrow$ \Call{restore}{$check$} 
    \State $stepstogo \coloneqq \max(capo-oldcapo-1,0)$ \Comment{need one less additional step since stage values are saved}
    \State \Return $restoredind$
  \EndFunction
\end{algorithmic}
\label{alg:mrevolve}
\end{algorithm}

We have implemented both Algorithms \ref{alg:revolve} and \ref{alg:mrevolve}
under the \texttt{TSTrajectory} class in \texttt{PETSc}, which provides two
critical API functions: \texttt{TSTrajectorySet()} and
\texttt{TSTrajectoryGet()}. The former function wraps \texttt{revolveForward}()
in \texttt{forwardSweep}(). The latter function wraps all the operations before
\texttt{adjointStep} in the \texttt{for} loop (Lines 7--10 in Algorithms
\ref{alg:revolve} and Lines 7--8 in Algorithms \ref{alg:mrevolve} ). This design
is suitable for preserving the established workflow of the ODE solvers so that
the influence on other interacting components such as \texttt{TSAdapt}
(time-step adaptivity class) and \texttt{TSMonitor} (time-step monitor class) is
minimized.

\texttt{PETSc} uses a redistributed
package\footnote{\url{https://bitbucket.org/caidao22/pkg-revolve.git}} that
contains a C wrapper of the original C++ implementation of \texttt{Revolve}.
Users can pass the parameters needed by \texttt{Revolve} through command line
options at runtime. \texttt{PETSc} provides additional command line options that
allow users to monitor the checkpointing process.

By design, \texttt{PETSc} manages the manipulation of checkpoints. The core data
structure is a stack  with push and pop operations, which is used to conduct the
actions decided by the checkpointing scheduler. Deep copy between the working
data and the checkpoints is achieved with the \texttt{PetscViewer} class. The
data can be encapsulated in either sequential or parallel distributed vectors.

Besides the offline checkpointing scheme, \texttt{PETSc} supports online
checkpointing and multistage checkpointing schemes, which are also provided by
the \texttt{Revolve} package. The proposed modification has been applied to
these schemes as well. Apart from memory, other storage media such as disk can
also be considered for storing checkpoints in binary format. For parallel file
systems, which are common on high-performance computing clusters, the
\texttt{PetscViewer} class can use MPI-IO to achieve high-performance parallel
I/O.

\subsection{Using \texttt{CAMS} in a discrete adjoint ODE solver}

Motivated by our experience with incorporating \texttt{Revolve} in
\texttt{PETSc} and the difficulties in handling the workflow, we design the main
interface function of \texttt{CAMS} to be \textbf{idempotent} and simplify its
output for better usability. The \texttt{CAMS} library is publicly available at
\url{https://github.com/caidao22/pkg-cams}. It provides both C and Python APIs.

To use \texttt{CAMS}, users first need to call the function
\begin{lstlisting}[style=CStyle]
offline_cams_create(int m,int s,int l,int stifflyaccurate);
\end{lstlisting}
to create a \texttt{CAMS} object and specify the number of time steps, the
number of allowable checkpointing units, the number of stages, and a flag that
indicates whether the integration method is stiffly accurate. At  creation time,
the dynamic programming algorithms presented in Sections
\ref{sec:stifflyaccurate} and \ref{sec:gmm} are executed, generating tables for
fast query access later on. A side benefit of dynamic programming is that when
$Prob_{\texttt{multistage}}(m,s)$ is solved, the solutions to the subproblems
$Prob_{\texttt{multistage}}(i,j)\ \forall 0<i<m,0<j<s $ become available.

Then users can query for the position and the type of next checkpoint by calling
\begin{lstlisting}[style=CStyle]
offline_cams(int lastcheckpointstep,int lastcheckpointtype,int s,int m,int l,
             int *nextcheckpointstep,int *nextcheckpointtype);
\end{lstlisting}
with negligible overhead.

The output of this function provides the minimal information necessary to guide
a checkpointing schedule. By design, it can be called repeatedly and return the
same output if provided the same input. This idempotence feature, however, is
not available in \texttt{Revolve} because the output of \texttt{Revolve} also
depends on hidden internal states. Redundant calls to the \texttt{Revolve}
function could be detrimental. In contrast, \texttt{CAMS} allows more \textbf{flexible}
and \textbf{robust} design of the workflow for the adjoint computation. Figure
\ref{fig:cams_workflow} illustrates how \texttt{CAMS} can be used in a typical
adjoint ODE solver. The core of the workflow is to determine whether a
checkpoint should be stored at the end of a time step, either during the forward
run or during the recomputation stage in the adjoint run. The other operations
such as restoring a checkpoint, recomputing from a restored state, or deleting an
unneeded checkpoint are intuitively straightforward. 
\begin{figure}[h]
  \centering
  \includegraphics[width=0.9\textwidth]{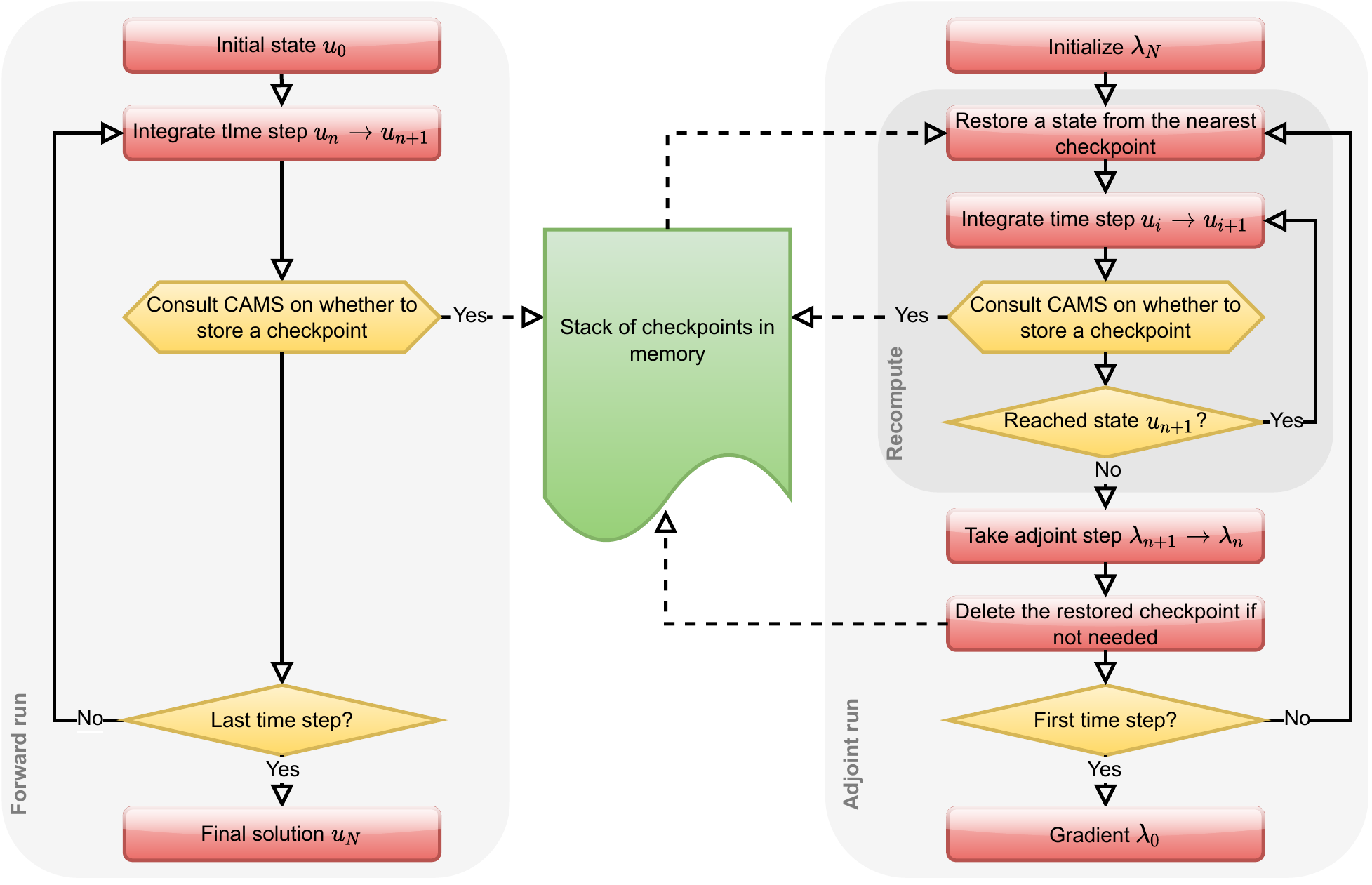}
  \caption{Workflow for the adjoint computation using \texttt{CAMS}.}
  \label{fig:cams_workflow}
\end{figure}

A detailed description of the algorithm for reversing a sequence of $m$ time
steps with \texttt{CAMS} is given in Algorithm \ref{alg:cams}. This algorithm
has also been implemented as an option in the \texttt{TSTrajectory} class in
\texttt{PETSc} with evident ease (compare Algorithm \ref{alg:cams} with
Algorithm \ref{alg:revolve} and \ref{alg:mrevolve}).
\begin{algorithm}[ht]
  \small
  \caption{Proposed adjoint computation for multistage time integration using \texttt{CAMS}.}
\begin{algorithmic}[1]
  \State $nextcptstep,nextcptype \leftarrow$ \Call{camsAPI} {$-1,-1,s,m,\ell$}
  \State Decide if the initial state should be stored according to $nextcptype$
  \State \Call{forwardSweep}{$0,m,state$}
  \State $adjstate \leftarrow$ \Call{adjointStep}{$adjstate$}
  \For{$i \coloneqq m-1$ to $1$}
    \State Restore the closest checkpoint and get $restoreind,restoredtype$
    \State $nextcpstep,nextcptype
    \leftarrow$ \Call{camsAPI} {$restoredind,restoredtype,s,m,\ell$}
    \State \Call{forwardSweep}{$restoredind,i-restoredind,state$}
    \State $adjstate \leftarrow$ \Call{adjointStep}{$adjstate$}
    \State Remove the checkpoint if no longer needed 
  \EndFor
  \Function{{forwardSweep}}{$ind,n,state$}
  \For{$i \coloneqq ind$ to $ind+n-1$}
    \State $state \leftarrow$ \Call{forwardStep}{$state$}
    \If{ $i = nextcpstep$}
        \State $nextcpstep,nextcptype \leftarrow$ \Call{camsAPI} {$nextcpstep,nextcptype,s,m,\ell$}
        \State Store the checkpoint and update $s$
    \EndIf
  \EndFor
  \EndFunction
\end{algorithmic}
\label{alg:cams}
\end{algorithm}

\section{Experimental Results}

To evaluate our algorithms for their performance and applicability to practical
applications at large scale, we ran experiments on the Cori supercomputer at
NERSC. In all the cases, we used $32$ Intel Knights Landing (KNL) nodes in cache
mode, $64$ cores per node, with each MPI process bound to one core. DDR memory
was used for storing checkpoints.

As a benchmark, we consider a PDE-constrained optimization problem for which an
adjoint method is used to calculate the gradient. The objective is to minimize
the discrepancy between the simulated result and observation data obtained from
a reference solution:
\begin{equation}
  \mathop{\text{minimize}}_{\bU_0} \| \bU(t_f) - \bU^{ob}(t_f)\|^2
\end{equation}
subject to the Gray--Scott equations \cite{hundsdorfer2007}
\begin{equation}
  \begin{aligned}
    \dot{\mathbf{u}} = D_1 \nabla^2 \mathbf{u} - \mathbf{u} \mathbf{v}^2 +
\gamma(1 -\mathbf{u}) \\
    \dot{\mathbf{v}} = D_2 \nabla^2 \mathbf{v} + \mathbf{u} \mathbf{v}^2 -
(\gamma + \kappa)\mathbf{v}, \\
  \end{aligned}
  \label{eq:diffusionreaction}
\end{equation}
where $\bU = [\mathbf{u} \; \mathbf{v}]^T$ is the PDE solution vector. A
reference solution is generated from the initial condition
\begin{equation}
  \mathbf{u}_0 = 1 - 2 \mathbf{v}_0, \quad
  \mathbf{v}_0  = 
  \begin{cases}
    \sin^2{(4\pi x}) \cos^2{(4\pi y)}/4,\; \forall x,y \in [1.0,1.5] \\
    0 \; \text{otherwise}
    \end{cases}
\end{equation}
and set as observed data. Solving the optimization problem implies recovering
the initial condition from the observations. This example is a simplified
inverse problem but fully represents the computational complexities and
sophistication in large-scale adjoint computations for time-dependent nonlinear
problems. 

Following the method of lines, the PDE is discretized in space with a centered
finite-difference scheme, generating a system of ODEs that is solved by using an
adjoint-capable time integrator in \texttt{PETSc} with a fixed step size. The
step size used is $1.0$ for implicit time integration and $0.001$ for explicit
time integration due to stability restrictions. For spatial discretization, the
computational domain $\Omega \in [0,2]^2$ is divided into a uniform grid of size
$128 \times 128$. The nonlinear system that arises at each time step is solved
with Newton's method. Both the linear and the transposed linear systems are
solved by using GMRES and a geometric algebraic multigrid preconditioner
following the same parameter settings described in \cite{zhang2021tsadjoint}. 

The base performance for adjoint computations with three selective time-stepping
schemes is presented in Table \ref{tab:diffusiont_reaction}. For implicit
time-stepping schemes, the reverse sweep takes much less time than the forward
sweep does, since only transposed linear systems need to be solved in the reverse
sweep whereas nonlinear systems need to be solved in the forward sweep. One
recomputation takes about $4$ times more than a backward step. Therefore,
savings in recomputations can lead to a dramatic reduction in total running time.
For explicit time-stepping schemes, to avoid the expensive cost of forming the
Jacobian matrix, we use a matrix-free technique to replace the transposed
Jacobian-vector product with analytically derived expressions. As a result, the
cost of one recomputation is less than the cost of a backward step.
\begin{table}
  \small
  \caption{Base performance for one adjoint computation (one forward sweep followed by one reverse sweep). In the tests with implicit schemes, the grid size is $2,048 \times 2,048$. For explicit schemes, the grid size is $16,384 \times 16,384$. Time is in seconds.}
  \label{tab:diffusiont_reaction}
  \centering
  \begin{tabular}{ c r r r}
    \toprule
    \multirowcell{2}{Time-stepping \\ schemes}  & \multirowcell{2}{Wall time for \\ forward sweep} &  \multirowcell{2}{Wall time for \\ reverse sweep} & \multirowcell{2}{Total time } \\
    & & & \\
    \midrule
    Backward Euler & 39.69 & 11.15  & 50.84  \\
    Crank--Nicolson  & 39.20  & 12.93 & 52.13  \\
    Runge--Kutta 4 & 1.56 & 2.10 & 3.66 \\
  \bottomrule
  \end{tabular}
\end{table}

Figure \ref{fig:mrevolve} shows the theoretical predictions regarding the
additional recomputations in the adjoint computation and the real recomputation
overhead in terms of CPU wall time. We see that the runtime decreases as the
allowable memory for checkpointing ($number\ of\ checkpointing\ units \times checkpoining\ unit\ size$) increases for the two time-stepping schemes
considered, backward Euler and Crank--Nicolson, both of which are stiffly
accurate. This decrease in runtime is expected because the number of additional
recomputations decreases as the number of allowable checkpointing units
increases. The best performance achieved by modified \texttt{Revolve} is
approximately $2.2$ times better than that by \texttt{Revolve}. \texttt{CAMS}
performs the best for all the cases. We  note that despite being noisy, the experimental results match the theoretical
predictions well. The timing variability is likely due to 
two aspects: one is that the cost of solving the implicit system varies across
the steps, which is expected for nonlinear systems; the other is
the run-to-run variations \cite{ChunduriHPMOCK17} on the KNL systems.
\begin{figure}[ht]
  \centering
  \subfloat[Theoretical predictions]{
    \includegraphics[width=0.48\textwidth]{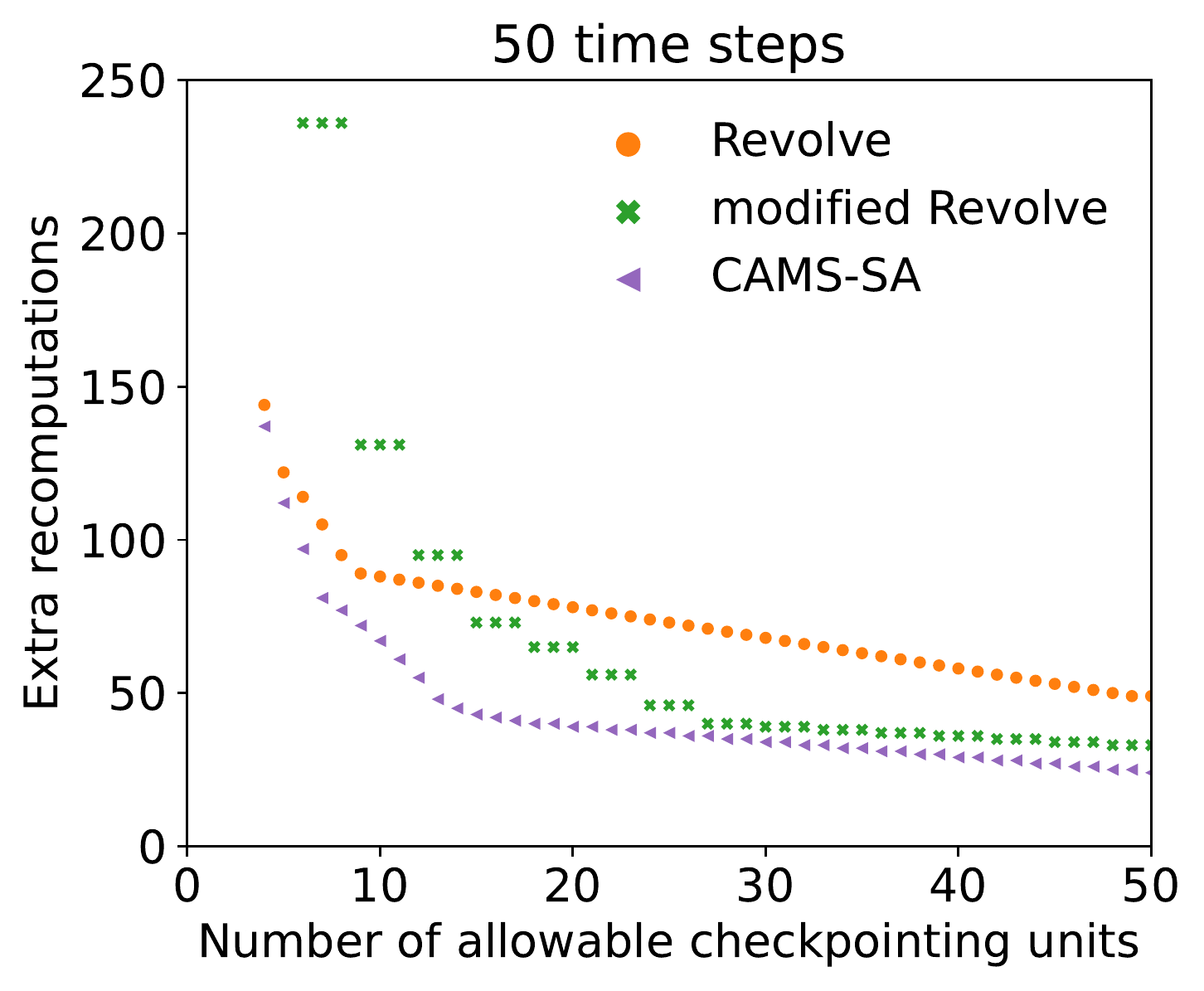}
  } \\
  \subfloat[Results with Crank--Nicolson]{
    \centering
    \includegraphics[width=0.48\textwidth]{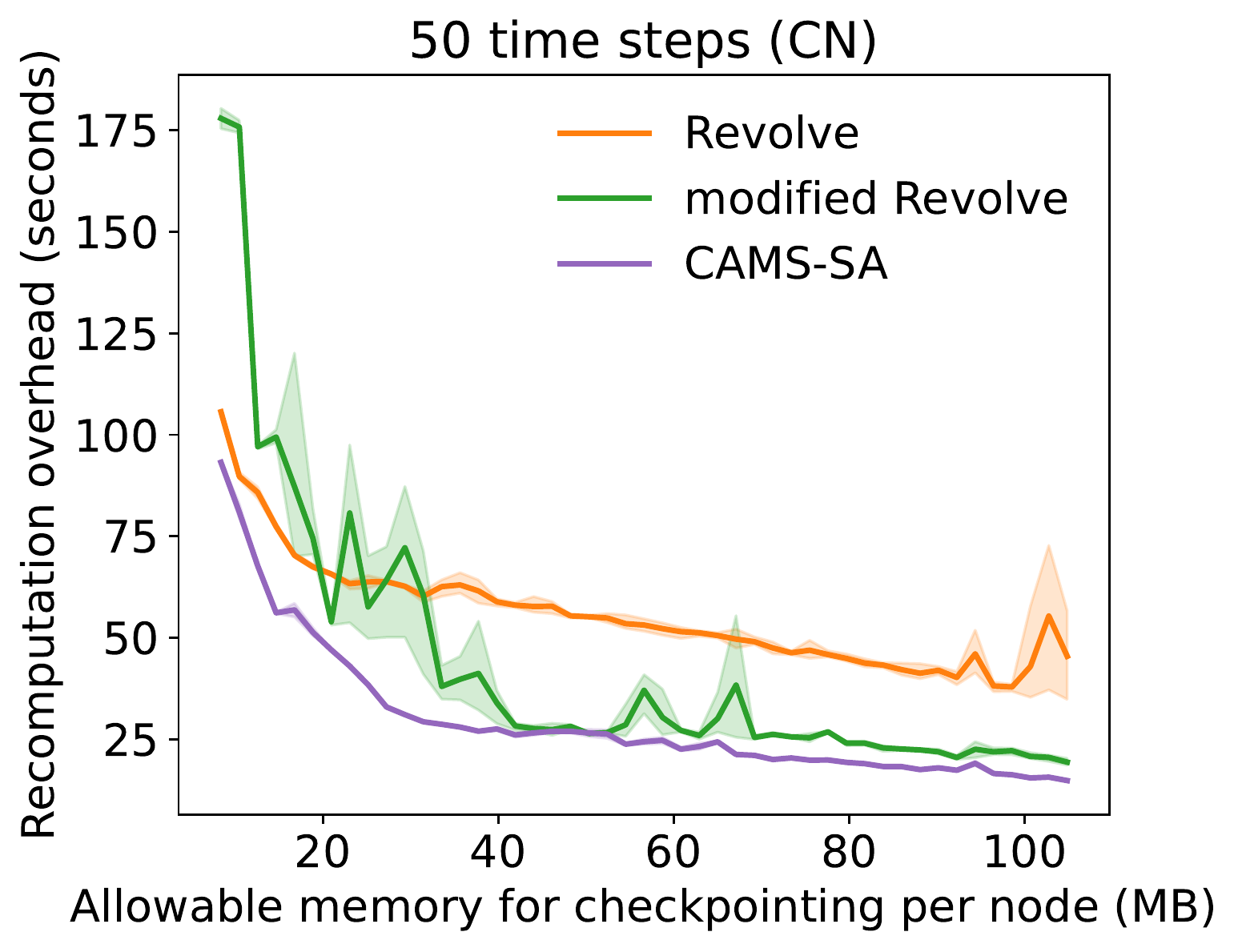}
  }
  \subfloat[Results with backward Euler]{
    \centering
    \includegraphics[width=0.48\textwidth]{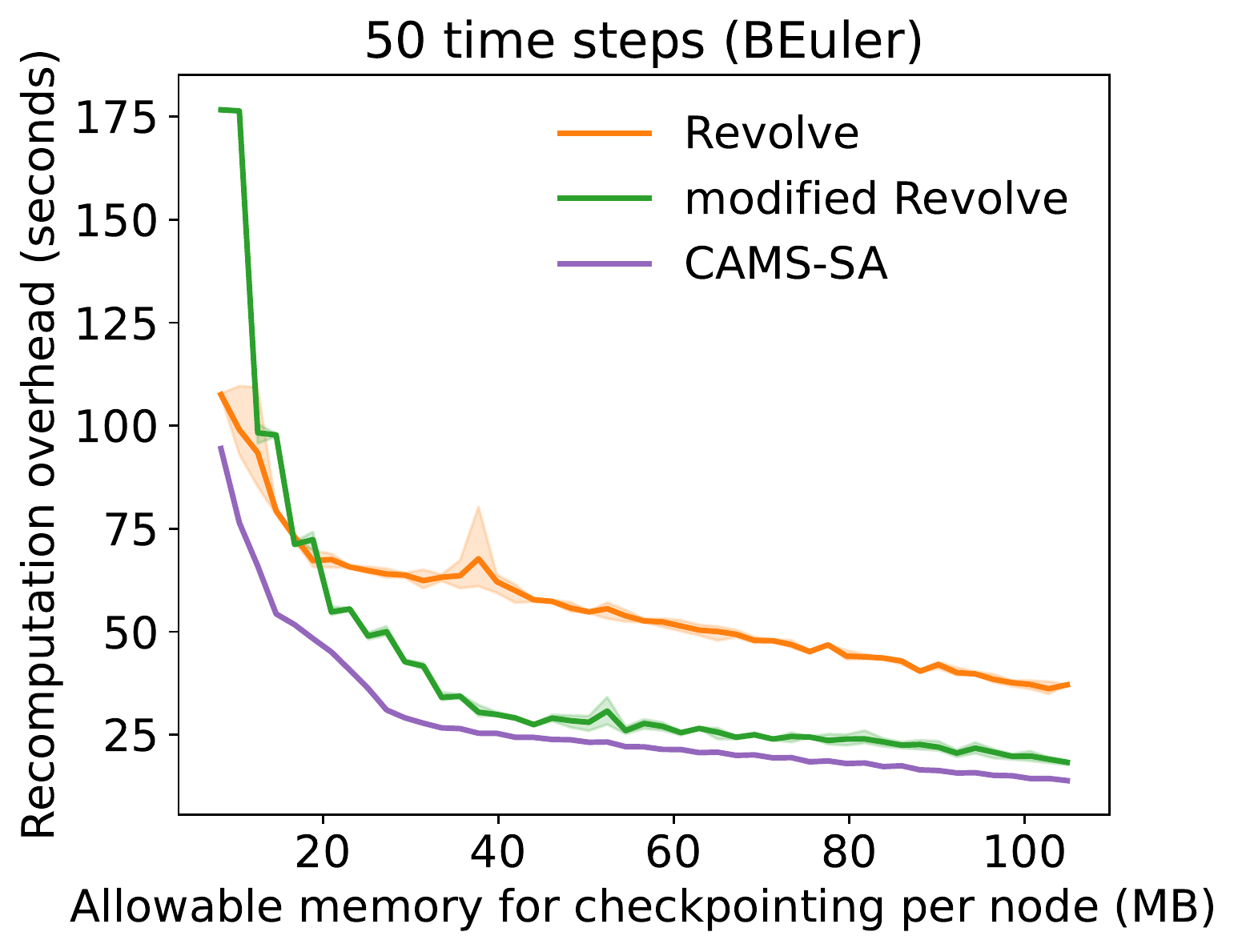}
  }
  \caption{Theoretical predictions and performance for \texttt{Revolve},
  modified \texttt{Revolve}, and \texttt{CAMS-SA}. In the experiments, the
  Crank--Nicolson method and backward Euler are used for time integration. Both
  methods can be viewed as a two-stage stiffly accurate method, where the second
  stage is the same as the solution. All the runs are executed on Cori with
  $2,048$ MPI processes ($32$ nodes) and are repeated three times to mitigate the influence of
  run-to-run variations.}
  \label{fig:mrevolve}
\end{figure}

Figure \ref{fig:mrevolve_rk} shows the result for the Runge--Kutta 4 method.
Modified \texttt{Revolve} outperforms \texttt{Revolve} when the allowable memory
for checkpointing per node is larger than roughly $6$ GB, which is close to the
theoretical prediction. We note that the performance of \texttt{CAMS-GEN}
slightly deviates from the theoretical prediction. This deviation is because the cost of
the memory movement accounts for a significant part of the recomputation
overhead, which can be partially explained by the fact that a KNL node has a low
DRAM bandwidth of $90GB/s$ but powerful cores with large caches and wide SIMD
units. However, \texttt{Revolve} is designed to minimize the number of times
saving a state to a checkpoint as a secondary objective \cite{Griewank2000} and
thus is less affected by the noticeable memory overhead than is
\texttt{CAMS-GEN}.
\begin{figure}[ht]
  \subfloat[Theoretical predictions]{
    \centering
    \includegraphics[width=0.46\textwidth]{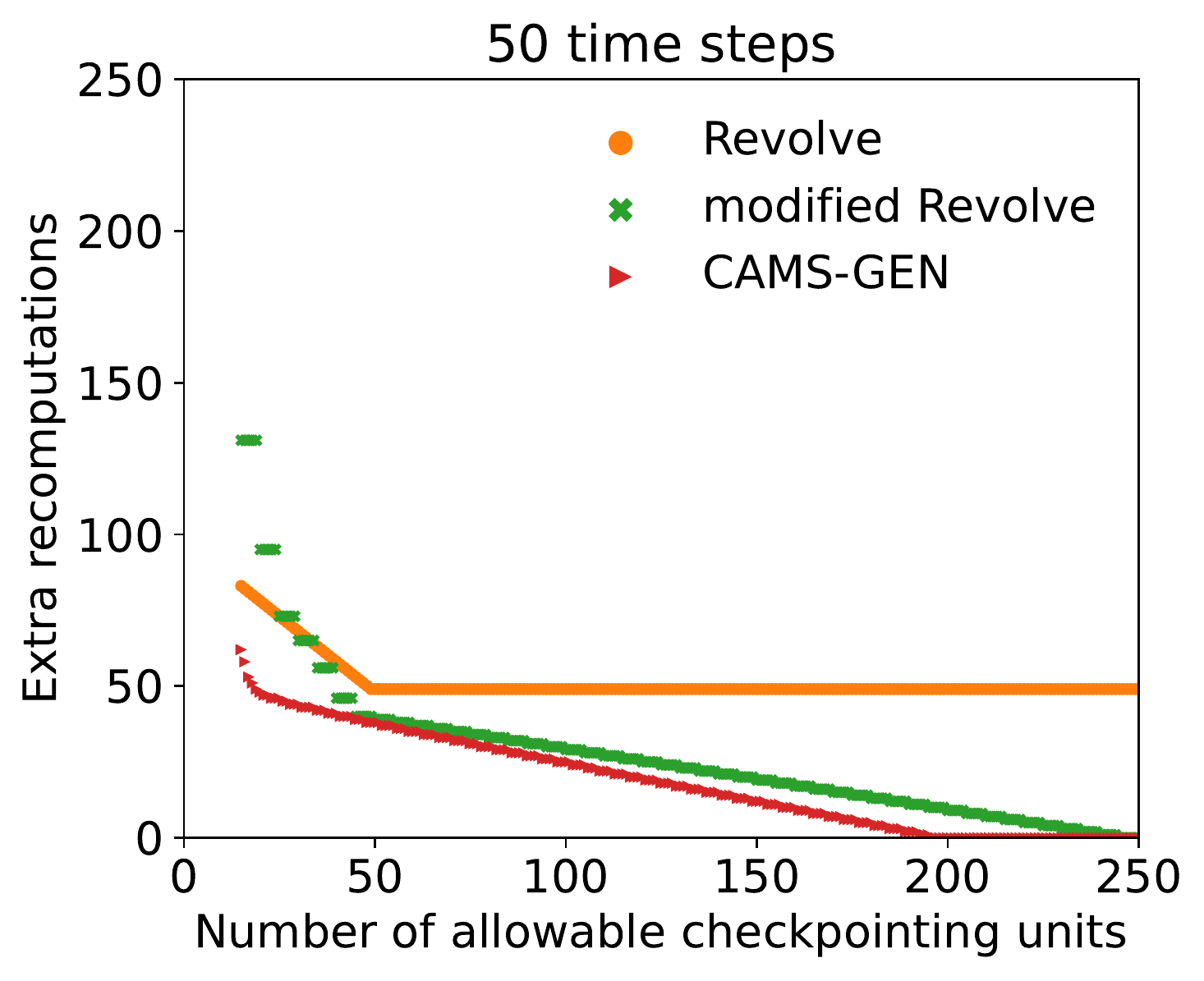}
  }
  \subfloat[Results with Runge--Kutta 4]{
    \centering
    \includegraphics[width=0.48\textwidth]{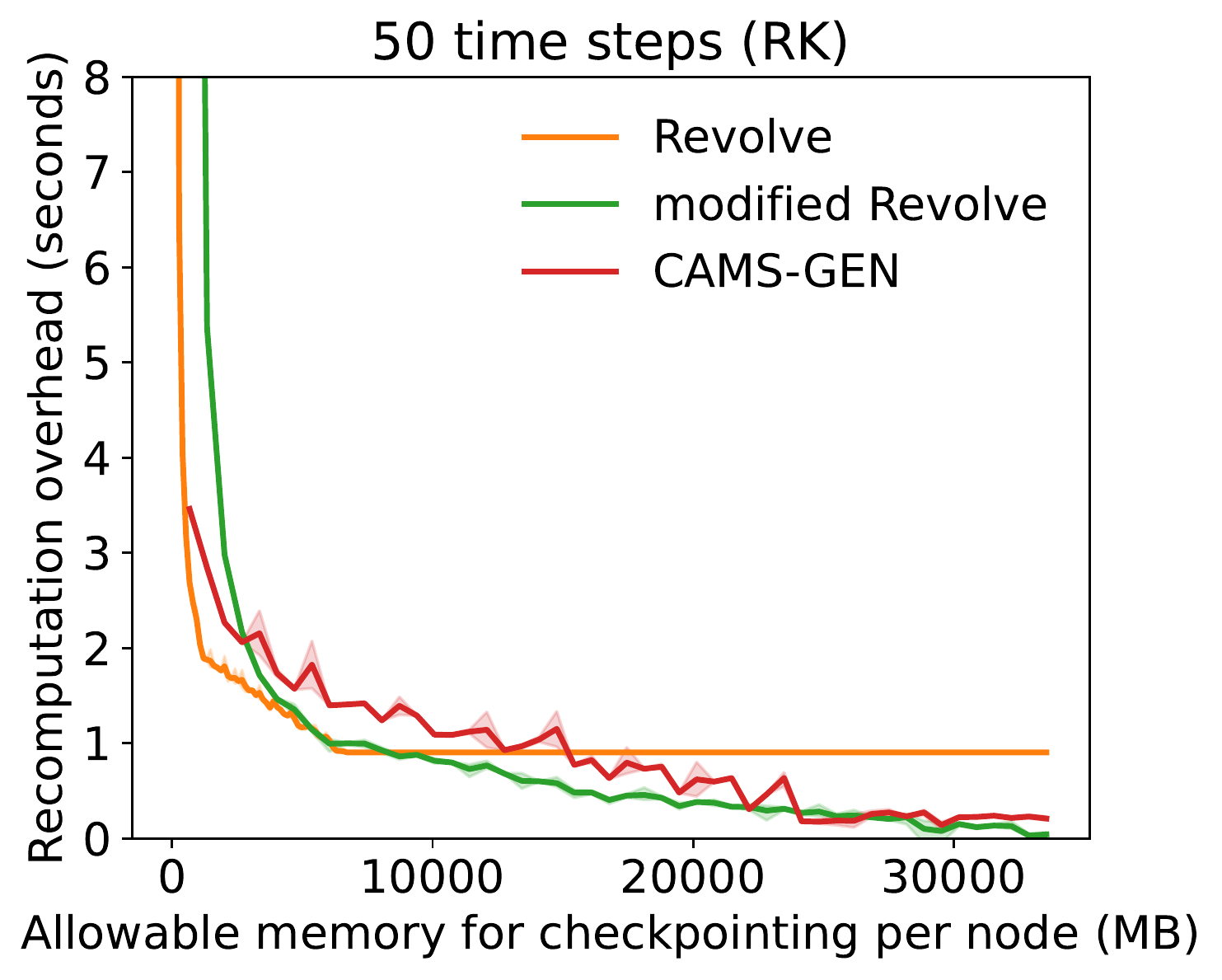}
  }
  \caption{Theoretical predictions and performance for \texttt{Revolve},
  modified \texttt{Revolve}, and \texttt{CAMS-GEN}. In the experiments, the
  Runge--Kutta 4 method is used for time integration. All the runs are executed
  on Cori with $2,048$ MPI processes and are repeated three times to mitigate
  the influence of run-to-run variations.}
  \label{fig:mrevolve_rk}
\end{figure}

\section{Conclusion}

The classical \texttt{Revolve} algorithm provides an optimal solution to the
checkpointing problem for adjoint computation in many scientific computations
when the temporal integration is abstracted at the level of time steps. 
When directly applied to multistage time-stepping schemes, however, it may not yield
optimal performance.

In this paper we have presented new algorithms that minimize the number of
recomputations under the assumption that the stage values of a multistage scheme
can be stored and the stage vectors are of the same size as the solution vector.
By extending from \texttt{Revolve} and redefining the checkpoint content, we
derived the modified \texttt{Revolve} algorithm that provides better performance
for a small number of time steps. Based on dynamic programming, we proposed the
\texttt{CAMS} algorithm, which proves to be optimal for computing the discrete
adjoint of multistage time-stepping schemes. We also developed a variant of
\texttt{CAMS} that takes advantage of the special property of stiffly accurate
time-stepping schemes. The performance has been studied both theoretically and
numerically. The results on representative cases show that our algorithms can
deliver up to $2$ times speedup compared with \texttt{Revolve} and do not need
recomputation when memory is sufficient.

In addition, the usage of \texttt{CAMS} is tailored to the workflow of practical
ODE solvers. We also propose a new approach for integrating the algorithms
introduced in this work into existing scientific libraries. In our approach the
solver does not require using our algorithms as a centralized controller over
the entire workflow, as proposed in the design of classical \texttt{Revolve};
thus our approach is less intrusive to the application codes. The proposed
algorithms have been successfully incorporated into the \texttt{PETSc} library.
For future work we will extend \texttt{CAMS} for online checkpointing problems
and investigate new algorithms that take memory access overhead into account.

\section*{Acknowledgments} We thank Jed Brown and Mark Adams for providing
 allocations on the Cori supercomputer.

\bibliography{paper.bib}
\bibliographystyle{ACM-Reference-Format}

 \begin{center}
	\scriptsize \framebox{\parbox{4in}{Government License (will be removed at publication):
			The submitted manuscript has been created by UChicago Argonne, LLC,
			Operator of Argonne National Laboratory (``Argonne").  Argonne, a
			U.S. Department of Energy Office of Science laboratory, is operated
			under Contract No. DE-AC02-06CH11357.  The U.S. Government retains for
			itself, and others acting on its behalf, a paid-up nonexclusive,
			irrevocable worldwide license in said article to reproduce, prepare
			derivative works, distribute copies to the public, and perform
			publicly and display publicly, by or on behalf of the Government. The Department of Energy will provide public access to these results of federally sponsored research in accordance with the DOE Public Access Plan. http://energy.gov/downloads/doe-public-access-plan.
}}
	\normalsize
\end{center}

\end{document}